\documentclass[aps,pra,letterpaper,twocolumn,nofootinbib,longbibliography]{revtex4-1}

\usepackage{color}
\usepackage{amsmath}
\usepackage{amsfonts}
\usepackage{amsthm}
\usepackage{amssymb}
\usepackage{dsfont}
\usepackage{graphicx}
\usepackage[caption=false]{subfig}
\usepackage{tikz}
\usepackage{hyperref}
\usepackage{color}
\usepackage{bm}
\usepackage{verbatim}
\usepackage{amscd}
\usepackage{setspace}
\usepackage{enumerate}
\usepackage[normalem]{ulem}

\newcommand{\calD}{\mathcal{D}}

\newcommand{\zz}{\mathbb Z}

\newcommand{\calM}{\mathcal M}
\newcommand{\calG}{\mathcal G}
\newcommand{\calB}{\mathcal B}
\newcommand{\calH}{\mathcal H}
\newcommand{\calS}{\mathcal S}
\newcommand{\calC}{\mathcal C}
\newcommand{\calL}{\mathcal L}
\newcommand{\calZ}{\mathcal Z}
\newcommand{\calK}{\mathcal K}

\renewcommand{\phi}{\varphi}

\newcommand{\norm}[1]{\left\lVert#1\right\rVert}

\newcommand{\Tr}{\operatorname {Tr}}

\newcommand{\bra}[1]{\ensuremath{\left\langle#1\right|}}
\newcommand{\ket}[1]{\ensuremath{\left|#1\right\rangle}}

\newcommand{\interior}[1]{%
	{\kern0pt#1}^{\mathrm{o}}%
}

\theoremstyle{plain}
\newtheorem{thm}{Theorem}
\theoremstyle{plain}
\newtheorem{lem}{Lemma}

\theoremstyle{definition}
\newtheorem{defn}{Definition}

\theoremstyle{remark}

\graphicspath{{./SPSCQMfigs/}}

\begin{document}
\title{Symmetry-protected self-correcting quantum memories}
\author{Sam Roberts and Stephen D. Bartlett}%
\affiliation{Centre for Engineered Quantum Systems, School of Physics, The University of Sydney, Sydney, NSW 2006, Australia}%

\date{20 August 2020}

\begin{abstract} 
A self-correcting quantum memory can store and protect quantum information for a time that increases without bound with the system size, without the need for active error correction.  We demonstrate that symmetry can lead to self-correction in 3D spin lattice models.  In particular, we investigate codes given by 2D symmetry-enriched topological (SET) phases that appear naturally on the boundary of 3D symmetry-protected topological (SPT) phases. We find that while conventional onsite symmetries are not sufficient to allow for self-correction in commuting Hamiltonian models of this form, a generalized type of symmetry known as a 1-form symmetry is enough to guarantee self-correction.  We illustrate this fact with the 3D `cluster state' model from the theory of quantum computing. This model is a self-correcting memory, where information is encoded in a 2D SET ordered phase on the boundary that is protected by the thermally stable SPT ordering of the bulk. We also investigate the gauge color code in this context. Finally, noting that a 1-form symmetry is a very strong constraint, we argue that topologically ordered systems can possess \textit{emergent} 1-form symmetries, i.e., models where the symmetry appears naturally, without needing to be enforced externally. 
\end{abstract}

\maketitle

\section{Introduction}

Quantum error correcting codes can be used to protect information in a noisy quantum computer.  While most quantum codes require complex active error correction procedures to be performed at regular intervals, it is theoretically possible for a code to be \emph{self-correcting}~\cite{Dennis,alicki2010thermal,BrownRev}.  That is, the energetics of a self-correcting quantum memory (SCQM) can suppress errors for a time that increases without bound in the system size, without the need for active control.  Such a memory is typically envisioned as a many-body spin system with a degenerate ground space. Quantum information can then be stored in its degenerate ground space for an arbitrarily long time provided that the system is large enough and the temperature is below some critical value. 

In seeking candidate models for self-correction, inspiration has been drawn from recent advances in our understanding of topologically ordered spin lattice models. The simplest example of a two-dimensional topologically ordered model is Kitaev's toric code \cite{Kitaev2003}, one of the most studied and pursued quantum error correcting codes. With active error correction, the toric code has a lifetime that grows exponentially with the number of qubits. However it is not self-correcting, as without active error correction the lifetime of encoded information is independent of the number of qubits. On the other hand, the four-dimensional generalization of the toric code \cite{Dennis} provides a canonical example of a self-correcting quantum memory. 

Encouraged by the capabilities of the 4D toric code, there has been a substantial effort to find self-correcting quantum memories that meet more physically realistic constraints and, in particular, exist in three or fewer spatial dimensions.  A number of no-go results make this search very challenging~\cite{BravTerh, PastYosh, YoshidaTopo, BravyiHaah, CardPoul,  CanLongRange,Komar2016}. While there has been considerable progress with proposals that attempt to circumvent these constraints in various ways~\cite{Haah,BravyiHaah,CanLongRange,Welded1,Welded2,ToricBoson,Pedrocchi,Brellcode,BrownEntropy}, none have yet provided a complete answer to the problem. 

Symmetry can provide new directions in the search for self-correcting quantum memories, as the landscape of ordered spin lattice models becomes even richer when one considers the interplay of symmetry and topology.  If a global symmetry is imposed on a model, a system can develop new quantum phases under the protection of this symmetry.   The properties that distinguish such symmetry-protected phases from more conventional phases persist only when these symmetries are not broken.  This has led to new types of phases protected by symmetry, including symmetry-protected topological (SPT) phases \cite{WenTopo,GroupCoho,chen2011classification,chen2011complete,schuch2011classifying,you2018subsystem} (phases with no intrinsic topological order) and symmetry-enriched topological (SET) \cite{WittenSPT,DefectsGauging,SET1,SET2,Fidkowski17,Wang13,Burnell14,Teo15,Mesaros13,Hermele14,Tarantino16,Vishwanath13,BulkBoundary1,BulkBoudnary2} phases (those including both intrinsic topological order and symmetry).  These phases have found many applications in quantum computing~\cite{darmawan2012,miller2016,bombin2010topological,YoshidaCCSPT,DBcompPhases,Miy10,else2012symmetryPRL,else2012symmetry,raussendorf2018computationally,bartlett2010haldanePRL,renes2013,MillerMiyake15,YoshidaHigher,NWMBQC,williamson2015symmetry,barkeshli2013twist,webster2018locality,brown2017poking,cong2017universal,ThermalSPT,zhu2018quantum,bartlett2017robust}. 

In this paper, we show that such phases can support self-correcting quantum memories in three-dimensions, provided an appropriate symmetry is enforced.  We argue that the generic presence of point-like excitations in commuting Hamiltonian models protected by an onsite symmetry precludes thermal stability (mirroring the instability of the 2D toric code), and so we are naturally led to consider higher-form symmetries.  Models with higher-form symmetries have excitations that are higher-dimensional objects, such as strings or membranes, rather than point-like excitations that are typical in models with onsite symmetries.  With such exotic excitations, we can seek models with the type of energetics believed to be needed for self-correction.  Focussing on models with symmetries that are not spontaneously broken, we consider models that have an SPT ordered bulk.  We then give two examples of 3D models that are self-correcting when a 1-form symmetry is enforced. The first example is based on the 3D `cluster state' model of Raussendorf, Bravyi and Harrington (RBH)~\cite{RBH}; this model with a 1-form symmetry has a bulk that remains SPT-ordered at non-zero temperature~\cite{ThermalSPT}.  We show that a self-correcting quantum memory can be encoded in a 2D SET boundary of this 3D model, and is protected by the thermally-stable SPT ordering of the bulk.  The second example is based on the 3D gauge color code~\cite{Bombin15}, which is conjectured to be self-correcting; we show that a commuting variant of this model is self-correcting when subject to a 1-form symmetry. 

Finally, we consider whether 1-form symmetries that lead to self-correction can be \emph{emergent}, rather than enforced. We say that a symmetry is emergent if the low-energy effective theory of a model strictly obeys this symmetry, rather than being required explicitly in the microscopic model. The analogy here is to the charge-parity symmetry that emerges in the effective anyon theory that describes the low-energy theory of many topologically ordered models, such as the toric code; such symmetries need not be externally enforced, as they are intrinsic to the model and stable under perturbations.  We give evidence that the 1-form symmetry used in the 3D gauge color code example may be emergent, arising as a result of emergent charge-parity symmetries on topologically-ordered codimension-1 submanifolds of the 3D bulk.  In the gauge color code, this symmetry is the `color flux conservation' identified by Bombin~\cite{SingleShot}. 

We would like to emphasise upfront an important subtlety in defining a symmetry-protected self-correcting quantum memory.  Enforcing symmetries can be extremely powerful, and along with potentially providing protection against errors, a poor choice in symmetry may be so strong as to render the system useless as a quantum memory.  In particular, one must be careful that the symmetry still allows for the encoding of logical information and implementation of logical operators using `local moves', i.e., sequences of local, symmetric operators. This requirement of a symmetry-protected SCQM will rule out some choices of strong symmetries.  For example, in the case of topological stabilizer codes, this removes the possibility of enforcing the entire stabilizer group as the symmetry (or for example, all of the vertex terms of a 3D toric code). We will revisit this subtle issue along with other rules in more detail in Sec.~\ref{sec:Thermal}. 

The paper is structured as follows.  In Sec.~\ref{sec:Background}, we review self-correction and the conditions required for it, as well as phases of matter protected by symmetry.  We analyse the effect of coupling symmetry-protected models to a thermal bath in Sec.~\ref{sec:Thermal}, and argue that onsite symmetries are insufficient to offer thermal stability of a symmetry-protected phase.  In Sec.~\ref{sec:RBH} we present our first example of a self-correcting quantum memory protected by a higher (1-form) symmetry:  a thermally-stable 3D SPT-ordered model with a protected 2D SET-ordered boundary.  A second example, based on the 3D gauge color code, is analyzed in Sec.~\ref{sec:GCC}.  We discuss the possibility of such 1-form symmetries being emergent in 3D topological models in Sec.~\ref{secEmergentHigherForm}, based around the gauge color code.  We discuss some implications of these results and open questions in Sec.~\ref{sec:Discussion}.

\section{Background}
\label{sec:Background}

In this section we briefly review self correcting quantum memories, topological phases with symmetry, and finally discuss how symmetries may play a role in self-correction.

\subsection{Self-correcting quantum memories}
The requirements of a self-correcting quantum memory have been formalized in the so-called `Caltech rules'~\cite{CanLongRange,Brellcode} (also see Ref.~\cite{BrownRev} for a review).  Specifically, a self-correcting quantum memory (SCQM) in $d$ spatial dimensions is a quantum many body spin system with the following four properties: (i) the Hilbert space consists of a finite density of finite-dimensional spins in $d$ spatial dimensions; (ii) the Hamiltonian $H$ has local terms with bounded strength and range, such that each spin is in the support of only a constant number of terms; (iii) the ground space of $H$ is degenerate (in the large size limit) such that a qubit can be encoded in the ground space and that this ground space is perturbatively stable; (iv) the lifetime of the stored information after coupling the system to a thermal bath must grow without bound in the system size. Typically, it is required that the lifetime grows exponentially in the system size, however there are situations where polynomial growth may be sufficient. Another desirable feature for a practical SCQM is the existence of an efficient decoder: a classical algorithm that corrects for errors in the system that have accrued over time. 

While the four-dimensional toric code meets all of the above requirements, there is currently no model that has been shown to meet these conditions in three-dimensions or fewer. The search for such a model has been encumbered by an assortment of no-go results for models consisting of commuting Pauli terms, known as stabilizer models~\cite{BravTerh,PastYosh,YoshidaTopo,BravyiHaah, CardPoul, CanLongRange}. These no-go results are typically centered around the idea that a SCQM must have a macroscopic energy barrier, meaning any sequence of errors that are locally implemented must incur an energy cost that diverges with the size of the system. (Note we will define the energy barrier more concretely in the following subsection.)  If a code has a macroscopic energy barrier then, naively, one may expect that logical faults can be (Boltzmann) suppressed by increasing the system size. This is indeed part of the puzzle, as it has been shown that a diverging energy barrier is necessary but not sufficient for self-correction for commuting Pauli Hamiltonians \cite{Temme2014,Temme2015} and abelian quantum doubles \cite{Komar2016}. (In particular, this rules out any codes based on entropic error suppression such as that of Brown \textit{et al.}~\cite{BrownEntropy}.) 

As such, any self-correcting quantum memory should be free of string-like (one-dimensional) logical operators, as these codes have a constant energy barrier. This holds since the restriction of a string-like logical to some region will commute with all terms in that region, and potentially only violate local terms near the boundary of the string. Therefore, to build up a logical fault (i.e., a logical string operator), one only needs to violate a constant number of terms, costing a constant amount of energy. This immediately rules out all 2D stabilizer codes \cite{BravTerh}, and 3D stabilizer Hamiltonians that have translationally invariant terms and a ground space degeneracy that is independent of system size (the so-called STS models of Yoshida \cite{YoshidaTopo}). Quantum codes in 3D that are free of string-like logicals have been investigated by Haah \cite{Haah,BravyiHaah} and Michnicki \cite{Welded1, Welded2}, however they do not achieve a memory time that is unbounded (with the size of the system) for a fixed temperature.

One class of proposals seeks to couple a 2D topologically ordered model, such as the toric code, to a 3D theory with long range interactions with the goal of confining the anyonic excitations. For example, excitations in the toric code can be coupled to the modes of a 3D bosonic bath \cite{ToricBoson,Pedrocchi,CanLongRange} such that anyonic excitations experience long range interactions. This coupling can result in a strong suppression of anyon pair production via a diverging chemical potential, and a confinement in excitation pairs leading to self-correcting behaviour. A complication with this approach is that the bulk generically requires fine tuning, and the chemical potential can become finite upon a generic perturbation~\cite{CanLongRange}. Such models are not self-correcting under generic perturbations.

Finally, while the search for self-correcting quantum memories has primarily focussed on stabilizer codes, \emph{subsystem codes}~\cite{PoulinSubsystem, 3DBaconShor} are a promising direction because many of the no-go theorems described above do not directly apply.  Briefly, a subsystem code is a stabilizer code where some of the logical qubits are chosen not to be used for encoding, and instead are left as redundant gauge degrees of freedom.  For the purposes of quantum memories, the use of subsystem codes and gauge qubits offers much more flexibility in selecting a Hamiltonian for the code, and the spectral requirements of the model for self-correction are potentially more relaxed.  The 3D gauge color code~\cite{Bombin15} is an example of a topological subsystem code with a variety of remarkable properties, including a fault-tolerant universal set of gates via a technique known as gauge fixing, and the ability to perform error-correction with only a single round of measurements. This later property is known as single-shot error correction~\cite{SingleShot} and arises from a special type of confinement of errors during the measurement step. It is conjectured in Ref.~\cite{Bombin15} that the 3D gauge color code is self-correcting.

\subsubsection{Thermalization and memory time}

The central question for a candidate self-correcting quantum memory is how long the encoded information can undergo thermal evolution while still being recoverable. For a self-correcting quantum memory, this time should grow with the system size provided the temperature is sufficiently low. In this section, we briefly review thermalization and motivate the energy barrier as a useful tool to diagnose the memory time.

The standard approach to modelling thermalization of a many body system is to couple the system to a thermal bosonic bath. Let $H_\text{sys}$ be the Hamiltonian describing the quantum memory of interest, and let $H_{\text{bath}}$ be a Hamiltonian for the bosonic bath. Thermalization is modelled by evolution under the following Hamiltonian 
\begin{equation}\label{eqCoupleBath}
H_{\text{full}} = H_{\text{sys}} +H_{\text{bath}} +  \lambda \sum_{\alpha} S_{\alpha} \otimes B_{\alpha},
\end{equation}
where $S_{\alpha} \otimes B_{\alpha}$ describe the system-bath interactions, $S_{\alpha}$ is a local operator acting on the system side, $B_{\alpha}$ is an operator acting on the bath side, and $\alpha$ is an arbitrary index. It is assumed that the coupling parameter is small, $|\lambda|\ll 1$.

Suppose that the state is initialized in a ground state $\rho(0)$ of $H_{\text{sys}}$. As the system is coupled to the thermal bath, after some time $t$ the system evolves to a noisy state $\rho(t)$.  Due the nature of the coupling, described by local coupling operators $S_{\alpha} \otimes B_{\alpha}$, errors are introduced to the system in a local way, and so the time evolution of the state $\rho(t)$ must be described by a local sequence of operations. One can give a precise description of this process using a perturbation theory analysis, such as a master equation approach like the well-known Davies formalism~\cite{Davies1,Davies2} which we review in Appendix~\ref{sec:Davies}.

For a self-correcting quantum memory, we wish to be able to recover the state $\rho(0)$ from $\rho(t)$ after some time $t$ using a single final round of error correction. Error correction consists of two steps, firstly a sequence of measurements is performed on the noisy state $\rho(t)$ to obtain an error syndrome, then a recovery map is performed that depends on the syndrome (the measurement outcomes). The net action of the syndrome measurement and recovery map can be condensed into a map $\Phi_{ec}: \calH \rightarrow \calH$, where $\calH$ is the Hilbert space of the memory system. For a fixed error rate $\epsilon$, we can define the memory time $\tau_{\text{mem}}$ as the maximum $t$ for which the inequality
\begin{equation}
\norm{\Phi_{\text{ec}} (\rho(t)) - \rho(0)}_1 \leq \epsilon
\end{equation}
is satisfied.

An upper bound to the memory time, is the mixing time $\tau_{\text{mix}}$, which is the time taken for $\rho(t)$ to be $\epsilon$ close to the Gibbs state (for some fixed $\epsilon$). This bound holds since once the system has thermalized to the Gibbs state, the system retains no information about the initial state. However, the memory time can be substantially less than the mixing time (as, for example, with the 3D toric code)~\cite{Temme2014}, and so this mixing time does not in general give us a tight bound on the memory time.  Instead, a useful proxy for determining the memory lifetime of a SCQM is the energy barrier, since a growing energy barrier is necessary in many cases to achieve self-correction. In the following subsection we define this quantity. 

\subsubsection{Energy barrier}

If we cannot recover the logical information after some time $t$, then we say that a logical fault has occurred. The coupling to the bath can lead to a logical fault if a sequence of local errors from the system-bath coupling results in a logical operator (or an operator near to a logical operator). Due to the locality of the coupling between the system and bath (in Eq.~(\ref{eqCoupleBath})), errors are introduced to the memory in a local way. There is an energy cost associated with any such process, which is directly related to the probability of such a process occurring when coupled to a bath at temperature $T$.  We now define this \emph{energy barrier} precisely.

We first define a local decomposition of a logical operator. In this paper we restrict to stabilizer Hamiltonians, however the energy barrier can similarly be defined for any commuting projector Hamiltonian. Let $H_{\calS} = -\sum_{i}h_i$ be a stabilizer Hamiltonian (i.e., each local term is a Pauli operator, and all terms mutually commute), and $\overline{l}$ a Pauli logical operator. A local decomposition of $\overline{l}$ is a sequence of Pauli operators $\calD(\overline{l}) = \{l^{(k)} ~|~ k = 1, \ldots N\}$ such that $l^{(1)}=I$ and $l^{(N)} = \overline{l}$, and $l^{(k)}$ and $l^{(k+1)}$ differ only by a local (constant range) operator.

For any ground state $\ket{\psi_0}$ of $H_{\calS}$, the state $l^{(k)}\ket{\psi_0}$ is also an eigenstate of $H_{\calS}$ (for each $k$) with energy $E^{(k)}$. We can use this to define the energy barrier $\Delta$ for a logical fault. Namely, the energy barrier for the local decomposition $\calD(\overline{l})$ is defined as 
\begin{equation}
\Delta_{\calD(\overline{l})} = \max_k (E^{(k)} - E_0),
\end{equation}
where $E_0$ is the ground space energy. The energy barrier for a logical fault in $H_{\calS}$ is defined as 
\begin{equation}\label{eqEnergyBarrier}
\Delta = \min_{\overline{l}, \calD(\overline{l})} \Delta_{\calD(\overline{l})}.
\end{equation}

In other words, the energy barrier for a logical fault is the smallest energy barrier of any logical operator, minimized over all local decompositions. Intuitively, the energy barrier should be large in order to suppress logical faults from occurring. 

The expectation for many models is that below some critical temperature the memory time will grow exponentially in the energy barrier
\begin{equation}
\tau_\text{{mem}} \sim e^{\beta \Delta}
\end{equation}
which is known as the \textit{Arrhenius law}. This relationship is observed to hold for many models such as the classical 2D Ising model and 4D toric code, but does not hold in general, (for instance in models when entropic effects are significant \cite{Haah,BravyiHaah,Welded1, Welded2}). Indeed for stabilizer Hamiltonians, an energy barrier that grows with the size of the system is a necessary condition (although not sufficient) for self-correction \cite{Temme2014, Temme2015}.

\subsubsection{Dimensionality of excitations and self-correction}
We conclude this subsection with a comment regarding the crucial role of the dimensionality of excitations in the feasibility of self-correction.  The conventional wisdom is that deconfined point-like excitations are an obstruction to self-correction, as harmful errors can be introduced with a low energy cost due to excitations that are free to propagate.  For models with higher dimensional excitations, the energy cost to growing and moving these excitations can be large, such that logical errors are suppressed. 

The properties of excitations and their dimensions for a given system can often be understood in terms of its symmetries.  As we will see in Sec.~\ref{sec:Thermal}, systems with global onsite symmetries have point-like excitations that are free to propagate, and therefore such symmetries do not offer any extra stability. This motivates the consideration of more general subsystem symmetries beyond the global onsite case. Higher-form symmetries are a family of symmetries that generalise the conventional global onsite symmetry. Excitations in systems with higher-form symmetries form higher-dimensional objects, and so their importance in the context of self-correction becomes apparent. 

\subsection{Topological phases with symmetry}
Quantum phases of matter are characterised by their ground state properties.  Two gapped local Hamiltonians are said to belong to the same phase if they are connected by a one-parameter continuous family of local Hamiltonians without closing the gap. When symmetry is at play, the classification becomes richer, as all Hamiltonians in the family must respect the symmetry. In particular, it is possible that two Hamiltonians that are equivalent in the absence of symmetry, become inequivalent when the symmetry is enforced. This leads to the notion of SPT and SET phases, which we now briefly define (see Ref.~\cite{WenTopo} for a detailed discussion).

Consider a lattice $\Lambda$ in $d$ dimensions with a $D$-dimensional spin placed at each site $i\in \Lambda$. We consider systems described by a gapped, local Hamiltonian $H = \sum_{X\subset \Lambda} h_X$. Here, `local' means that each term $h_X$ is supported on a set of spins $X$ with bounded diameter. We also assume the system has a symmetry described by a group $G$ with a unitary representation $S$. We say two gapped Hamiltonians $H_0$ and $H_1$ with symmetry $S(g)$, $g\in G$ belong to the same phase if there exists a continuous path of gapped, local Hamiltonians $H(s)$ $s\in [0,1]$ that all obey the symmetry $S(g)$ such that $H(0) = H_0$ and $H(1) = H_1$. 

For SPT and SET ordered systems, one commonly considers global symmetries $S(g)$ that act via an onsite fashion on the underlying degrees of freedom. The global action of these onsite symmetries $S(g)$ may be expressed as 
\begin{equation}\label{eqOnsiteSym}
S(g) = \bigotimes_{i \in \Lambda} u(g), \quad g \in G,
\end{equation}
where $u(g)$ is a local, site-independent representation of $G$. 

We will also consider a generalised class of global symmetries, known as higher-form symmetries, which have been recently of high interest in the condensed matter, high energy and quantum information communities \cite{kapustin2017higher,HigherForm1,YoshidaHigher,ThermalSPT,HigherForm2,lake2018higher}. These higher-form symmetries form a family of increasingly stringent constraints that generalize the onsite case, and this will be central in the discussion of the interplay of symmetry and self-correction. We introduce these symmetries in Sec.~\ref{secHigherForm}, and for the present discussion and the definitions of SPT and SET phases, the action of the symmetry $S(g)$ is left general.

\subsubsection{Symmetry protected topological phases} 

An SPT phase with symmetry $S(g)$ is defined as class of symmetric Hamiltonians which are equivalent under local symmetric transformations which do not close the gap and which are not in the same class as the trivial phase (a non-interacting spin model with a product ground state), but which \emph{are} in the same phase as the trivial model if the symmetry were not enforced.  Ground states of such models are short range entangled, meaning they can be mapped to a product state under a constant depth quantum circuit; however, such a circuit must break the symmetry. Key characteristics of such phases are the absence of anyonic excitations, and the absence of topology dependent ground space degeneracy. However, when defined on a lattice with boundary, these phases host protected modes localized on the boundary, meaning the boundary theory of an SPT phase must be either symmetry breaking, gapless, or topologically ordered (note that a topologically ordered boundary can only exist when the boundary has dimension $d \geq 2$). As such, these systems are typically regarded as having a trivial bulk, but exotic boundary theories. Some well known examples are the 1D cluster state and the spin-1 Haldane phase (with $\zz_2^2$ symmetry), both of which host degenerate boundary modes that transform as fractionalized versions of the symmetry. More generally the group cohomology models \cite{GroupCoho} provide a systematic way of constructing SPT ordered models.
  
\subsubsection{Symmetry enriched topological phases} 

An SET phase with symmetry $S(g)$ is defined by a Hamiltonian that is distinct from the trivial phase, even without any symmetry constraint. These topological phases can form distinct equivalence classes under the symmetry $S(g)$, and are referred to as SET phases. The key characteristics of such phases are the presence of anyonic excitations, and topology-dependent ground space degeneracy. These anyons can carry fractional numbers of the symmetry group, or may even be permuted under the symmetry action. Such anyon permuting symmetries can be used to define symmetry defects on the lattice, which can be thought of as localized and immobile quasiparticles that transform anyonic excitations when they are mutually braided. Some well known examples of SETs are found in Refs.~\cite{DefectsGauging,Teo15,Mesaros13,Hermele14,Tarantino16}, and a general framework is given by the symmetry-enriched string-nets of Refs.~\cite{SET1,SET2}. These SET phases fall into two categories.  The first category consists of non-anomalous SET phases.  These are standalone topological phases in $d$-dimensions with onsite symmetry $S(g)$ as in Eq.~(\ref{eqOnsiteSym}). Anyons may undergo transformations under the symmetry action $S(g)$. The second category consists of anomalous SET phases. These are $d$-dimensional topological phases with a symmetry action that cannot be realised in an onsite way on the degrees of freedom on the $d$-dimensional boundary. These anomalous phases appear only on the boundary of $(d{+}1)$-dimensional SPT phases.

It is conjectured that the topologically ordered boundary of an SPT phase with bulk onsite symmetry must always be anomalous. In particular, a wide class of 3-dimensional SPT phases can be classified by the group cohomology models~\cite{GroupCoho}, which are labelled by elements of the cohomology group $H^4(G,U(1))$. (See Refs.~\cite{vishwanath2013physics,wang2013boson,burnell2014exactly,kapustin2014symmetry} for examples of models outside this classification.) Moreover, in 2 dimensions, anyonic systems with discrete unitary symmetry $G$ (that does not permute the anyons) also have a label in $H^4(G,U(1))$ that classifies the anomalies~\cite{H4Obsturction} (see also~\cite{DefectsGauging}). The case $\omega=1$ (i.e., trivial) means that there is no anomaly, and $\omega \neq 1$ means the system is anomalous and cannot be realised in 2-dimensions in a standalone way with onsite symmetries~\cite{WittenSPT,Fidkowski17,Wang13,Burnell14,Vishwanath13,BulkBoundary1}. A conjecture of Ref.~\cite{BulkBoundary1} is that the gapped boundary topological theory of a group cohomology model must always have an anomaly $\omega \in H^4(G,U(1))$ that agrees with the label specifying the bulk SPT order. This kind of bulk-boundary correspondence was proved in Ref.~\cite{BulkBoudnary2} in the case that the symmetry group $G$ is abelian and does not permute the boundary anyons. Moreover, in Ref.~\cite{AnomalousBoundary}, a general procedure to extract a boundary anomaly label from a bulk SPT has been given, in agreement with the conjecture. 

\subsubsection{Higher-form symmetries}
\label{secHigherForm}

We will make use of a family of symmetries called higher-form symmetries~\cite{kapustin2017higher,HigherForm1, YoshidaHigher,ThermalSPT,HigherForm2,lake2018higher}, generalizing the onsite case. These symmetries have been of recent interest for several reasons, in particular, they provide a useful structure for error correction in quantum computation \cite{ThermalSPT}, have been used to construct new phases of matter \cite{YoshidaHigher}, and to understand topological phases from the symmetry breaking paradigm \cite{HigherForm1,lake2018higher}. 

A $q$-form symmetry (for some $q\in \{0,1,...,D{-}1\}$) is given by a symmetry operator associated with every closed codimension-$q$ submanifold of the lattice; these operators are written as $S_{\calM}(g)$ where $\calM$ is a closed codimension-$q$ submanifold of $\Lambda$ and $g \in G$.  On these codimension-$q$ submanifolds, the action of the symmetry operators takes an onsite form:  for $g\in G$ and a codimension-$q$ submanifold $\calM$, the symmetry operator is
\begin{equation}\label{eqHigherFormSym}
S_{\calM}(g) = \prod_{i \in \calM} u(g), \quad g \in G
\end{equation}
where the product runs over all sites $i$ of the submanifold $\calM$, and $u(g)$ is a local, site-independent representation of $G$.
That is, higher-form symmetries can be thought of as being onsite symmetries on lower dimensional submanifolds.  For systems with a boundary, we only require that the submanifolds on which the higher-form symmetries are supported are closed relative to the boundary of the lattice. In other words, the manifold $\calM$ on which the symmetry is supported may have a boundary on the boundary of the lattice $\Lambda$, i.e. $\partial \calM \subset \partial \Lambda$. 

A key feature of systems with $q$-form symmetries is that symmetric excitations must form $q$-dimensional objects. Of particular interest in this paper will be 1-form symmetries in 3-dimensional systems, which are the next weakest generalization (within the family of higher-form symmetries) of the conventional global onsite symmetry. Symmetry operators in such systems are supported on closed 2-dimensional surfaces, and excitations form closed 1-dimensional loop-like objects. In Sec.~\ref{sec:RBH} and Sec.~\ref{sec:GCC} we will look at two examples of self-correcting quantum memories protected by $\zz_2^2$ 1-form symmetries.

\subsubsection{Self-correction and topological order}
The relationship between self-correction and thermal stability is complex. Self-correction is a dynamic property of a system, whereas thermal stability is an equilibrium property.  In many previous investigations, various quantities have been used as proxies or indicators of self-correction, for instance, the existence of a nonzero temperature phase transition \cite{Brellcode,hastings2014self}, the presence of topological entanglement entropy in the Gibbs state \cite{castelnovo2008topological}, or the nontriviality of Gibbs ensemble in terms of circuit depth \cite{hastings2011topological}. Here, by `thermal stability' we specifically mean the presence of topological order in the thermal state, as determined by the minimal circuit depth to prepare, following Refs.~\cite{hastings2011topological,ThermalSPT}. While we do not yet have a general result connecting the thermal stability and memory time, we explore the connection between these two notions further through the example of the RBH model, by proving bulk thermal stability from the existence of a macroscopic energy barrier on the boundary. This type of bulk-boundary correspondence (at nonzero temperature) provides evidence in favour of a close relationship between thermal stability and self-correction.

\subsection{Symmetry constraints and quantum memories}
\label{sec:Thermal}

In this section, we consider what types of symmetric models may be worth investigating as potential self-correcting quantum memories.

An important condition that must be met by a symmetry-protected self-correcting quantum memory is that all logical operators can be implemented through a sequence of symmetric local moves. That is, all logical operators $\overline{l}$ admit a local decomposition $\calD(\overline{l}) = \{l^{(k)} ~|~ k = 1, \ldots N\}$, such that all $l^{(k)}$ are symmetric. If this condition is not met, then one cannot guarantee the existence of a symmetric, local encoding circuit. Furthermore, this condition implies that even in the presence of symmetry, the bath is capable of implementing all logical faults and that the logical information will eventually be thermalized. If such a condition is not met, one can construct `trivial' self-correcting models in which the symmetry is spontaneously broken, as explained below.

\subsubsection{No spontaneous symmetry breaking}

If we require our model to admit symmetric local decompositions of all logical operators, then the enforced symmetry $S(g)$ cannot be spontaneously broken. In a model where the symmetry is spontaneous broken, the ground space has less symmetry than the Hamiltonian, and this can render the model trivial as a memory by disallowing logical operator actions at all.  Different ground states will in general be in different eigenspaces of the symmetry operator, and thus enforcing the symmetry would be prohibit transitions between ground states.  In the case that the spontaneously broken symmetry is higher-form, enforcing it could remove some or all of the anyonic excitations from the model. 

The 3D toric code provides an illustrative example, where one can trivially obtain a self-correcting quantum memory by enforcing a $\zz_2$ 1-form symmetry that prevents any of the vertex terms from flipping. Enforcing the vertex and plaquette terms in a 2D toric code provides another trivial example of this phenomenon. These examples do not admit symmetric local decompositions of all logical operators. For this reason, we only consider models where the symmetry is not spontaneously broken, and SPT ordered systems provide a natural family of candidates.

\subsubsection{Onsite symmetries are insufficient for stability}\label{secOnsite}

In this section we argue that onsite symmetries are insufficient to promote a 2D topological quantum memory to be self-correcting, even if such a phase lives on the boundary of a 3D SPT model.  Our goal here is simply to motivate moving beyond onsite symmetries (to higher-form symmetries), not to rigorously rule out any role for onsite symmetries in the study of self-correction.  

In particular, consider the case where the full system is given by a commuting Hamiltonian with boundary, and that the protecting symmetry is abelian and onsite (with possibly an anomalous boundary action).  The excitations in such systems will be point-like, and their presence precludes the possibility of having thermally stable (symmetry-protected) topological order, as shown in Ref.~\cite{ThermalSPT}. This suggests that the boundary theory is also not thermally stable, and thus not self-correcting. Indeed, as we show in Appendix~\ref{sec:ThermalOnsite}, this is the case for the class of models where the boundary is an abelian twisted quantum double with a potentially anomalous boundary symmetry.  Specifically, we show that there is a constant (symmetric) energy barrier in this case.  Therefore we see that in the case of onsite (0-form) symmetries, the SPT ordered bulk offers no additional stability to the boundary theory. This motivates us to consider the boundaries of SPTs protected by 1-form (or other higher-form) symmetries.

\subsubsection{System-bath coupling with symmetry and the symmetric energy barrier}

Consider the system bath coupling of Eq.~(\ref{eqCoupleBath}) and a symmetry $S(g)$ (with $g\in G$ for some group $G$). If 
\begin{equation}\label{eqSBSym}
[H_{\text{full}}, S(g)] = 0,
\end{equation}
then all of the errors that are introduced due to interactions with the bath must be from symmetric processes that commute with $S(g)$. In particular, only excitations that can be created by symmetric thermal errors will be allowed and the symmetry is preserved throughout the dynamics. 

Under symmetric dynamics, we should only consider local decompositions of logical operators that commute with the symmetry when defining the energy barrier $\Delta$. If a local decomposition $\calD(\overline{l}) = \{l^{(k)} ~|~ k = 1, \ldots, N\}$ of a logical operator $\overline{l}$ is such that $[l^{(k)},S(g)] = 0$ for all $k$ and all $g\in G$, then we call $\calD(\overline{l})$ a symmetric local decomposition of $\overline{l}$. We label such symmetric local decompositions with symmetry $G$ by $\calD_G(\overline{l})$. Then the \emph{symmetric energy barrier} is defined as 
\begin{equation}\label{eqSymEnergyBarrier}
\Delta_G = \min_{\overline{l}, \calD_G(\overline{l})} \Delta_{\calD_G(\overline{l})}.
\end{equation}
Namely, it consists of the smallest energy barrier for any logical operator, where the cost is minimized over all symmetric local decompositions. For notational simplicity, we often omit the subscript $G$ as the symmetry is clear from context.

With the abundance of no-go results for self-correction in 2D and 3D stabilizer memories, the relevant question is whether one can achieve self-correction if the system bath coupling respects a symmetry. In particular, for a given model $H_{\calS}$, can a symmetry $S(g)$ be imposed such that $H_{\text{sys}}$ has a macroscopic symmetric energy barrier?

\section{Self-correction with a 1-form SPT phase}
\label{sec:RBH}

Our first example of a 3D self-correcting model in the presence of a 1-form symmetry is described by a commuting Hamiltonian based on the cluster-state model of Raussendorf, Bravyi, and Harrington (RBH)~\cite{RBH}. This model has been used in high-threshold schemes for fault-tolerant quantum computation~\cite{Rau06,RBH,TopoClusterComp}. In particular, the RBH model underpins the topological formulation of measurement-based quantum computation, where single qubit measurements are used to simulate the braiding of punctures in the 2D toric code.

The RBH model is an example of an SPT ordered system under 1-form symmetry, which is thermally stable~\cite{ThermalSPT}.  It contains no anyonic excitations in the bulk, however when defined on a lattice with a boundary, the boundary theory can be gapped, topologically ordered, and possesses point-like anyonic excitations. In particular, the boundary can be chosen to be described by a boundary Hamiltonian equivalent to the 2D surface code. Without any symmetry, the excitations of this 2D surface code phase are deconfined, and information encoded in this surface will thermalize in constant time in the absence of error-correction.  However, in the presence of symmetry, a natural question is whether the boundary code inherits any protection from the bulk SPT order. We will show that in the presence of 1-form symmetry, the bulk SPT order gives rise to confinement of boundary excitations and ultimately a macroscopic lifetime of boundary information. As such, this model provides a simple example of an anomalous SET phase on the boundary of a 3D higher-form SPT that is thermally stable, giving a self-correcting quantum memory.

We first define and present the bulk properties of this model. We then define some important boundaries of the model, including the anomalous toric code SET phase. Finally we present the global lattice and boundary conditions and discuss the resulting model as a quantum code and show that it results in a symmetry-protected SCQM.

\subsection{The RBH model -- bulk properties}

\subsubsection{The RBH bulk Hamiltonian}

In this subsection, we define the RBH model in the bulk.  Consider a 3D cubic lattice $\calL$. Label the set of all vertices, edges, faces and volumes of $\calL$ by $V$, $E$, $F$, $Q$. Similarly, to prepare ourselves for boundary conditions that are to be specified later, we label the interior vertices, edges, faces and volumes by $\interior{V}$, $\interior{E}$, $\interior{F}$, $\interior{Q}$, and $\interior{\calL}$ is the collection of all interior cells. For now we ignore any boundary conditions (meaning we consider only interior cells), and one may consider periodic boundary conditions until specified otherwise. We place a qubit on every face $f\in F$ and on every edge $e \in E$. We refer to qubits on faces as primal qubits, and qubits on edges as dual qubits.

The bulk Hamiltonian is a sum of commuting cluster terms
\begin{equation}\label{eqBulkHam}
H_{\interior{\calL}} = -\sum_{f \in \interior{F}} K_{f} -\sum_{e \in \interior{E}} K_{e},
\end{equation}
where each cluster term is a 5-body operator
\begin{equation}\label{eqClusterStateStabilizers}
K_{f} = X_f \prod_{e:e \subset f} Z_e, \qquad K_{e} = X_e \prod_{f: e \subset f} Z_f,
\end{equation}
and $X_v$ and $Z_v$ are the usual Pauli-$X$ and Pauli-$Z$ operators acting on the qubit $v$. These terms are depicted in Fig.~\ref{figClusterLattice}

We note that the terms in the Hamiltonian can be considered `dressed' terms of a simpler, trivial bulk model. In particular, we define the ``trivial model'' $H_{\interior{\calL}}^{(0)}$ to be a trivial paramagnet:
\begin{equation}\label{eqHamTriv}
H_{\interior{\calL}}^{(0)} = -\sum_{i \in \interior{E}\cup \interior{F}} X_i.
\end{equation}

One can see that these two models are equivalent up to a constant depth circuit
\begin{equation}
H_{\interior{\calL}} = UH_{\interior{\calL}}^{(0)}U^{\dagger},
\end{equation}
where $U$ is a product of controlled-$Z$ gates that act on all pairs of neighbouring qubits at sites $i$ and $j$ by
\begin{equation}\label{eqCZ}
CZ_{ij} = \exp \left( \frac{i\pi}{4}(1-Z_i)(1-Z_j) \right). 
\end{equation}
Indeed, let a face $f$ and an edge $e$ be referred to as neighbours if the edge is contained within the face $e \subset f$. Then $U$ is a product of controlled-$Z$ gates over all neighbouring sites
\begin{equation}\label{eqRBHcircuit}
U = \prod_{f \in \interior{F}} \prod_{e \subset f} CZ_{f e}.
\end{equation} 

From this we can see that the bulk Hamiltonian $H_{\interior{\calL}}$ is non-degenerate (since $H_{\interior{\calL}}^{(0)}$ is non-degenerate).

\begin{figure}
	\centering
	\subfloat[]{	\includegraphics[width=0.4\linewidth]{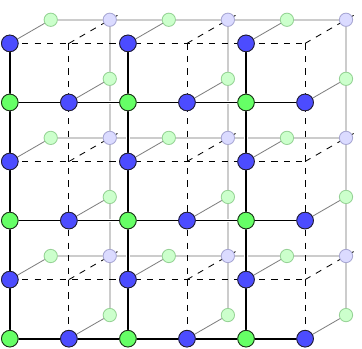}}	\label{figClusterBdrya}
	\quad
	\subfloat[]{	\includegraphics[width=0.4\linewidth,angle=-00]{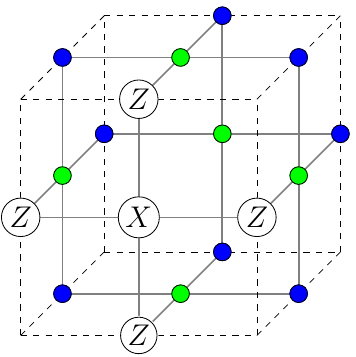}}	\label{figClusterBdryb}
	\caption{(a) A portion of the bulk lattice. Primal qubits are depicted in green, while dual qubits are depicted in blue. (b) A bulk cluster term $K_f$. In both figures, bold lines indicate nearest neighbour relations between qubits, while dashed lines indicate edges of the ambient cubic lattice.}%
	\label{figClusterLattice}
\end{figure}

\subsubsection{Bulk excitations without symmetry}

We now consider the excitations in the model in the absence of any symmetry considerations. In the bulk, all excitations can be created by products of Pauli-$Z$ operators applied to the ground state. Indeed, for any subset of edges $E' \subset \interior{E}$ or subset of faces $F'\subseteq \interior{F}$, the operator 
\begin{equation}\label{eqBulkEx}
Z(E',F') = \prod_{f \in F'} Z_f \prod_{e \in E'}Z_e
\end{equation}
anti-commutes with precisely the cluster terms $K_e$ and $K_f$ for which $e \in E'$ and $f \in F'$, and commutes with all remaining bulk terms. Moreover, all excitations can be reached in this way (as can be verified by considering the trivial model $H_{\interior{\calL}}^{(0)}$ and the local unitary $U$ of Eq.~(\ref{eqRBHcircuit})).  The energy cost for creating excitations at sites in $E' \cup F'$ with the operator $Z(E',F')$ is given by 
\begin{equation}
|E'\cup F'|\Delta_{\text{gap}}
\end{equation}
where $\Delta_{\text{gap}}=2$ is the energy gap. 

The bulk model is very simple due to its low-depth equivalence with the trivial paramagnet. Excitations can be locally created on any site by flipping a spin, they have no interaction with each other, and the energy cost of a general excitation is proportional to the number of flipped spins.

We refer to excitations supported on sites  $F'\subseteq \interior{F}$ as primal excitations, and excitations supported on sites $E' \subset \interior{E}$ as dual excitations.

\subsubsection{1-form symmetries} 

The model $H_{\interior{\calL}}$ has a $\zz_2^2$ 1-form symmetry, consisting of operators supported on closed 2-dimensional surfaces on each of the primal and dual sublattices. In particular, a generating set are given by vertex and cube operators (for dual and primal qubits, respectively), for each $ q \in Q$ and $v \in V$
\begin{equation}\label{eqRBH1formSymOps}
S_q = \prod_{f:f \subset q} X_f, \quad S_v = \prod_{e: v\subset e} X_e. 
\end{equation}

In the bulk, these operators are depicted in Fig.~\ref{figCluster1form}. Taking products of these operators gives rise to the $\zz_2^2$ 1-form symmetry 
\begin{equation}\label{eqSymG}
G = \langle S_v , S_q ~|~ v \in V, q \in Q\rangle.
\end{equation}
A general 1-form symmetry operator is generated by a product of Pauli $X$ operators on faces of (relative) 2-cycles of the lattice $\calL$ and edges dual to (relative) 2-cycles on the dual lattice. Comparing to the general expression of 1-form symmetries in Eq.~(\ref{eqHigherFormSym}), we note that the codimension 2 surfaces $\mathcal{M}$ are given by these (relative) 2-cycles and dual 2-cycles.

\begin{figure}
	\centering
	\subfloat[]{	\includegraphics[width=0.45\linewidth]{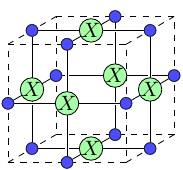}}
	\quad
	\subfloat[]{	\includegraphics[width=0.45\linewidth]{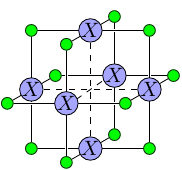}}
	\caption{Generators of the 1-form symmetry in the bulk. (a) A primal generator $S_q$. (b) A dual generator $S_v$. Thick lines denote neighbour relations, and dashed lines denote the cubic lattice.}
	\label{figCluster1form}
\end{figure}

One can easily check that these operators commute with both $H_{\interior{\calL}}$ and $H_{\interior{\calL}}^{(0)}$.

We briefly remark on lattices with boundaries. While the microscopic content of the 1-form operators may change near a boundary (as dictated by the boundary lattice geometry), the macroscopic and topological features remain unchanged -- namely that they are generated by operators supported on 2-dimensional surfaces that are closed relative the boundary. We will examine the boundaries in the following sections in analogy to the boundaries of topological phases, such as the surface code~\cite{bravyi1998quantum}, where microscopic details help us understand the topological properties of the boundary, but are themselves not the most important features.

It has been shown that under these symmetries the bulk model $H_{\interior{\calL}}$ belongs to a nontrivial SPT phase while the trivial bulk $H_{\interior{\calL}}^{(0)}$ belongs to the trivial phase. Moreover, this distinction persists to nonzero temperature, where $H_{\interior{\calL}}$ remains SPT ordered~\cite{ThermalSPT}. In particular, while the whole unitary $U$ commutes with the symmetry, the individual $CZ$ gates do not. In fact, there is no constant depth circuit with local gates that commute with the symmetry mapping the RBH model to the trivial model.

\subsubsection{Bulk excitations with 1-form symmetries}

We now consider what excitations are possible in the presence of the 1-form symmetry $G$. If we consider bulk excitations, then the excitation operator $Z(E', F')$ of Eq.~(\ref{eqBulkEx}) is symmetric if and only if both $E'$ is a 1-cycle (i.e., it has no boundary) and $F'$ is a 2-cocycle (meaning it is dual to a 1-cycle on the dual lattice -- where vertices are replaced with cubes, edges with faces, and so on). In other words, the only symmetric bulk excitations are formed by combinations of closed loop-like (i.e., 1-dimensional) objects, and we refer to them as loop excitations. We can further refer to loop excitations as either primal or dual if they are supported on sets of faces or edges, respectively.  

Both the primal and dual loop excitations have an energy cost proportional to their length, and are thus confined. This confinement leads to thermal stability of the model.

\subsection{Boundaries I: General considerations}\label{secBoundarySec}
To obtain degeneracy in the ground space we must consider a lattice with boundaries. The allowable boundary Hamiltonians are dictated by the symmetry action on the boundary, which in turn is determined by the boundary geometry. In addition to changing the ground space degeneracy of the model, the choice of boundary Hamiltonian may allow for different types of excitations to condense on them. By condense, we mean that an excitation can be absorbed on the boundary (and the reverse process is also possible, where excitations can be emitted from a boundary). In the following, we will consider four different types of symmetric gapped boundary Hamiltonians that each allow different excitations to condense on them. These boundaries will allow us to construct the Hamiltonian with a degenerate ground space (i.e. codespace) that is self-correcting under 1-form symmetry. 

We will first focus on a toric code boundary which will be used to encode information. We will then introduce other boundary types that can be combined with the toric code boundary to construct a code that allows for all logical operators to be implemented through a sequence of symmetric local moves (as required by the discussion in Sec.~\ref{sec:Thermal}).

\subsubsection{Boundary condensation}
Throughout this section, it will be useful to characterise boundaries in terms of the types of excitations that can condense on them. By boundary, we mean a combination of the choice of how to terminate the lattice, the symmetry appropriately defined on this lattice, and a Hamiltonian that commutes with the symmetry (we will see examples of these choices in the next subsection). We define a boundary as being primal-condensing or dual-condensing as follows.

\begin{defn}
	We refer to a boundary as primal-condensing (dual-condensing) if any primal (dual) loop excitation can be piecewise removed near the boundary using local, symmetric operations.
\end{defn}
A schematic depicting a dual-condensing boundary is shown in Fig.~\ref{figDualCondensable}. Importantly, for a boundary to be able to condense a general loop excitation, it must be capable of piecewise condensing it. This piecewise requirement is what makes the above definition nontrivial, as small loop excitations can always be condensed wholly, by contracting them to a point (which, however, is not true for loop excitations with nontrivial topology). Importantly, a boundary is primal-condensing (dual-condensing) if and only if primal (dual) string excitations can terminate on them in a symmetric way. For example, Fig.~\ref{figDualCondensable}~$(ii)$ depicts a dual loop excitation terminating on a dual-condensing boundary. Therefore, symmetric excitations only need to be closed loops modulo their respective primal/dual-condensing boundaries. 

Macroscopically, these types of boundaries can be thought of as analogues of the $X$ and $Z$-type boundaries of the surface code (often called rough and smooth)~\cite{bravyi1998quantum}. While a boundary may have its own set of excitations that are localised within it (and they may interact with bulk loop excitations), the definitions of primal-condensing and dual-condensing are independent of this. We now look at an important boundary that is both primal-condensing and dual-condensing.

\begin{figure}
	\centering
	\includegraphics[width=0.99\linewidth]{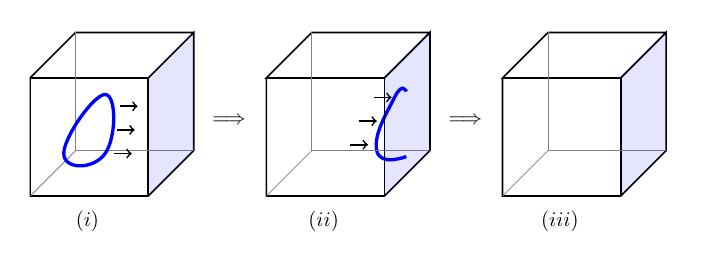}
	\caption{A dual-condensing boundary can absorb a dual loop excitation. $(i)$ A dual loop excitation in the bulk is depicted in blue, while the dual-condensing boundary is shaded light blue. $(ii)$ The loop is moved to the boundary, where part of it is absorbed. $(iii)$ The loop is fully absorbed.}
	\label{figDualCondensable}
\end{figure}

\subsection{Boundaries II: The toric code boundary}\label{ToricBoundaries}
As mentioned, the type of Hamiltonian that can be defined on the boundary is heavily constrained by the symmetry. We first consider boundary conditions that support a 2D toric code phase.  

We consider a lattice with one boundary component which we terminate with `toric code' boundary conditions (see Fig.~\ref{figTerms}). Namely, the cubic lattice is terminated on a smooth plane, such that there are boundary volumes, boundary faces, boundary edges, and boundary vertices, each having a lower number of incident cells (neighbours) compared to the bulk. We label the collection of all boundary volumes, faces, edges, and vertices by $\partial \calL$. We will fix the topology and geometry more precisely later, for this section we consider a lattice supported on a 3D half space, i.e., with coordinates $(x,y,x)$ satisfying $x \geq 0$, $-\infty \textless y \textless \infty$, $-\infty \textless z \textless \infty$, such that the boundary is on the $x=0$ plane. On the boundary, qubits are placed only on boundary edges, and not on boundary faces, as depicted in Fig.~\ref{figTerms}. We refer to these qubits as boundary qubits. (Note that we have constructed this boundary using dual qubits. This choice is arbitrary, and an analogous boundary exists that is comprised of primal qubits.)

\begin{figure}
	\centering
	\subfloat[]{	\includegraphics[width=0.42\linewidth]{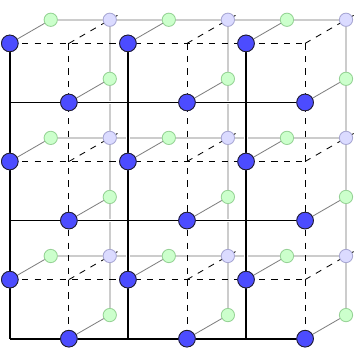}}	\label{figClusterBdryc}
	\quad
	\subfloat[]{	\includegraphics[width=0.42\linewidth]{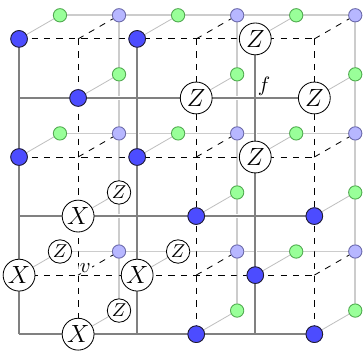}}
	\caption{(a) The boundary of the lattice consists only of dual qubits which are depicted in blue. Primal qubits on faces penetrating into the bulk are depicted in green. (b) The boundary terms $\overline{A}_v$ and $\overline{B}_f$.  In both figures, bold lines indicate nearest neighbour relations, while dashed lines indicate edges of the cubic lattice. The dashed lines on the boundary can be thought of as the edges of a toric code lattice.}
	\label{figTerms}
\end{figure}

For this geometry, we consider Hamiltonians of the form
\begin{equation}\label{eqRHam}
H = H_{\interior{\calL}} + H_{\partial \calL},
\end{equation}
where, $H_{\interior{\calL}} $ is the bulk Hamiltonian of Eq.~(\ref{eqBulkHam}) (which sums only over sites on the interior, meaning it contains only complete cluster terms) and $H_{\partial \calL}$ is a boundary Hamiltonian. A boundary Hamiltonian is in general any Hamiltonian with local terms acting near the boundary of the lattice $\partial \calL$ that commute with the symmetry (whose action we describe shortly). 

\subsubsection{Boundary degrees of freedom}
To determine what types of Hamiltonians $H_{\partial \calL}$ are possible on the boundary, we describe the boundary Hilbert space in terms of a more natural boundary algebra. We begin with the case $H_{\partial \calL} = 0$ such that $H = H_{\interior{\calL}}$ consists of all 5-body cluster terms of Eq.~(\ref{eqClusterStateStabilizers}). In this case there is an extensive degeneracy localised near the boundary: there is a qubit 'boundary degree of freedom' for every boundary edge (i.e. one for every $e\in E\cap \partial \calL$). It is important to distinguish between the qubits that belong to the boundary, and the degrees of freedom localised near the boundary that describe the ground space. Indeed, the operators that act on these degrees of freedom within the ground space of $H$ are not simply given by the Pauli operators acting on boundary qubits. That is, for some Pauli operator $P_e$ acting on $e\in E\cap \partial \calL$, we have $\Pi_0 P_e \Pi_0 \neq \Pi_0 P_e$ in general, where $\Pi_0$ is the ground space projector (in fact we have equality only if the bulk is a trivial paramagnet). The effective Pauli-$X$ and Pauli-$Z$ operators can be obtained by finding the dressed versions of these Pauli operators using the unitary in Eq.~(\ref{eqRBHcircuit}).

More explicitly, the effective Pauli-$X$ and Pauli-$Z$ operators for these boundary degrees of freedom are given by 
\begin{align}\label{eqClusterBdryDOF}
\tilde{X}_e := UX_eU^{\dagger} = X_e Z_{f(e)}, \quad  \tilde{Z}_e := UZ_eU^{\dagger} = Z_e,
\end{align}
where $U$ is a product of $CZ$ gates as in Eq~(\ref{eqRBHcircuit}) and ${f(e)}$ is the unique face $f\in \interior{F}$ such that $e \subset f$. These operators preserve the ground space (as they commute with all bulk cluster terms in $H_{\interior{\calL}}$) and act on the boundary degrees of freedom in the ground space as the usual Pauli spin operators. We will describe boundary degrees of freedom in terms of the boundary algebra generated by $\tilde{X}_e$, $\tilde{Z}_e$. We emphasise that the support of the boundary algebra is not strictly contained on the boundary qubits, as would be the case if the bulk Hamiltonian was trivial. This subtle difference between the boundary degrees of freedom and cut boundary qubits is important, as we will see.

\subsubsection{Symmetry action on the boundary}\label{SecSymBoundary}

The $\zz_2^2$ 1-form symmetry on a lattice with a boundary is again given by the group $G$ in Eqs.~(\ref{eqRBH1formSymOps}) and~(\ref{eqSymG}) (consisting of operators supported on 2-dimensional submanifolds that are closed relative the lattice boundary). The operators near the boundary are depicted in Fig.~\ref{figBoundarySyms}.

\begin{figure}
	\centering
	\subfloat[]{	\includegraphics[width=0.45\linewidth]{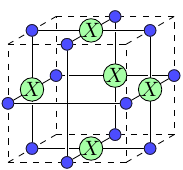}}
	\quad
	\subfloat[]{	\includegraphics[width=0.45\linewidth]{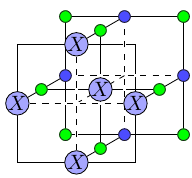}}
	\caption{Symmetry operators on the boundary (a) $S_q$ with $q\in \partial \calL$, (b) $S_v$ with $v\in \partial \calL$. Thick lines denote neighbour relations, and dashed lines denote the cubic lattice.}
	\label{figBoundarySyms}
\end{figure}

A general boundary Hamiltonian can be written in terms of operators from the boundary algebra. We must therefore analyse the action of the 1-form symmetry on the boundary algebra (to infer how the boundary degrees of freedom transform under the symmetry). First, we note that the operators of Eq.~(\ref{eqClusterBdryDOF}) are not themselves symmetric. Taking the boundary symmetry operators $S_v$ and $S_q$ with $v \in V\cap \partial \calL$, $q \in Q\cap \partial \calL$ (depicted in Fig.~\ref{figBoundarySyms}), for any $e \in E \cap \partial L$ we have (under conjugation)
\begin{align}
S_v : \tilde{X}_e &\mapsto \tilde{X}_e, \quad & \tilde{Z}_e &\mapsto (-1)^{\mathds{1}_e(v)}\tilde{Z}_e \label{eqBoundaryCommRel1} \\
S_q : \tilde{X}_e &\mapsto (-1)^{\mathds{1}_q(e)} \tilde{X}_e, \quad & \tilde{Z}_e &\mapsto \tilde{Z}_e, \label{eqBoundaryCommRel2}
\end{align}
where ${\mathds{1}_e(v)} = 1$ if $v \subset e$ and ${\mathds{1}_e(v)} = 0$ otherwise, and similarly ${\mathds{1}_q(e)} = 1$ if $e\subset q$ and ${\mathds{1}_q(e)} = 0$ otherwise.

From this we can write the action of the 1-form symmetry in the ground space of $H$ in terms of operators in the boundary algebra as follows. Define the following `dressed toric code' operators for every $v \in V\cap \partial \calL$ and every $f \in F \cap \partial \calL$:
\begin{equation}\label{eqDressedTC}
\overline{A}_v = \prod_{e\in \partial E: v \subset e}X_e \prod_{f : e \subset f} Z_f, \qquad \overline{B}_f = \prod_{e: e \subset f} Z_e,
\end{equation}
where $\partial E = E \cap \partial \calL$ is the set of boundary edges. Such operators are depicted in Fig.~\ref{figTerms}. They are dressed versions of the usual toric code operators 
\begin{equation}\label{eqToricTriv}
A_v = \prod_{e\in \partial E: v \subset e}X_e \qquad B_f = \prod_{e: e \subset f} Z_e,
\end{equation}
and can be obtained by conjugating them by the unitary of Eq.~(\ref{eqRBHcircuit}).

Now it can be verified from the (anti)commutation relations of Eqs.~(\ref{eqBoundaryCommRel1}-\ref{eqBoundaryCommRel1}) that the 1-form symmetry acts as  
\begin{align}\label{eqBoundaryAction}
S_v &\equiv \overline{A}_v&   &\forall~ v \in V \cap \partial \calL, \\
S_q &\equiv \overline{B}_{f(q)}&  &\forall ~ q \in Q \cap \partial \calL,
\end{align}
and as the identity otherwise. Here $f(q)$ is the unique face $f(q) = \partial q \cap \partial \calL$, and $\overline{A}_v$ and $\overline{B}_f$ are defined in Eq.~(\ref{eqDressedTC}). The equivalence $\equiv$ means that the two operators have the same action in the ground space. In other words, $S_v$ and $\overline{A}_v$ (resp. $S_q$ and $\overline{B}_{f(q)}$) have identical commutation relations with all boundary operators $\tilde{X}_e$ and $\tilde{Z}_e$ of Eq.~(\ref{eqClusterBdryDOF}), and therefore have equivalent action on the boundary degrees of freedom.

Thus we see that the SPT-ordered bulk requires the boundary theory to be nontrivial: in order to respect the symmetry, any boundary Hamiltonian must commute with the dressed toric code operators and therefore any trivial boundary Hamiltonian (e.g. trivial paramagnet) is ruled out. There are two further observations to make about the action of the symmetry on the boundary. Firstly, the symmetry is represented as a 1-form symmetry on the boundary degrees of freedom: i.e. $\overline{A}_v$ and $\overline{B}_{f}$ generate a symmetry group whose elements are supported on closed loops. Secondly, the supports of these symmetry operators are not strictly contained on the boundary qubits.

\subsubsection{Toric code boundary Hamiltonian}
In order to add a nontrivial Hamiltonian $H_{\partial \calL}$ to the boundary, it must be composed of terms that commute with $\overline{A}_v$ and $\overline{B}_f$ from Eq.~(\ref{eqDressedTC}). {As such, the canonical choice of boundary Hamiltonian has terms are given by $\overline{A}_v$ and $\overline{B}_f$.} This gives us the dressed toric code boundary
\begin{equation}\label{eqDressedToricCodeBound}
H_{\partial \calL} = - \sum_{v \in \partial V} \overline{A}_v -\sum_{f \in \partial F} \overline{B}_f,
\end{equation}
where $\partial V$ and $\partial F$ are the set of all boundary vertices and faces (respectively). Again, the terms of this Hamiltonian are depicted in Fig.~\ref{figTerms}.

\subsubsection{Toric code boundary excitations}
The toric code Hamiltonian introduces a new set of excitations on the boundary, that are interesting in themselves, but also interact nontrivially with bulk excitations. 

The boundary supports anyonic excitations that are free to propagate in the absence of any symmetry. Indeed, for a string $l \subseteq \partial E$ on the boundary, we can define the string operator $Z(l) = \prod_{e \in l} Z_e$. The string operator $Z(l)$ commutes with all Hamiltonian terms, apart from  vertex terms $\overline{A}_v$ with $v\in \partial l$ for which it anti-commutes with. We define flipped $\overline{A}_v$ terms as $e$-excitations, and string operators $Z(l)$ create these excitations. Similarly, we can define a dual-string operator $\tilde{X}(l') = \prod_{e \in l'} X_e \prod_{f \in l'^{\perp}} Z_f$ for a string $l' \subset \partial E$, which when applied to the ground space, creates $m$-excitations on the faces at the ends of $l'$. Here, $l^{\perp} = \{f \in \interior{F}: \partial f \cap l \neq \emptyset \}$ denotes the set of faces sitting just inside the boundary incident to the string $l$. At the endpoints of the string operator $\tilde{X}(l')$, $m$-excitations occur, as the plaquette operators $\overline{B}_f$ with $f$ on the ends of $l'$ anti-commute with $\tilde{X}(l')$, while all remaining terms commute. Examples of such operators are depicted in Fig.~\ref{figBoundaryLogicals}

\begin{figure}
	\centering
	\subfloat[]{	\includegraphics[width=0.4\linewidth]{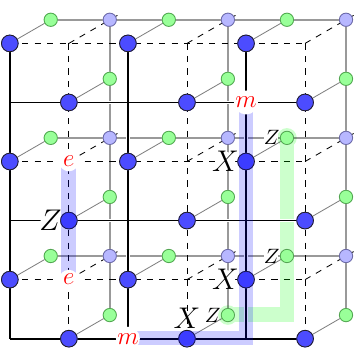}}\label{figBoundaryLogicalsa}
	\quad
	\subfloat[]{\includegraphics[width=0.4\linewidth]{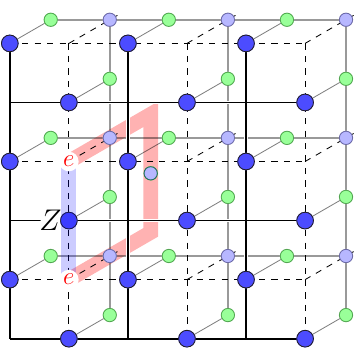}}\label{figBoundaryLogicalsb}
	\caption{(a) The $e$ and $m$ type excitations on the boundary of string and dual-string operators. For an $e$-excitation ($m$-excitation) to be symmetric they must be accompanied by a bulk dual (primal) string excitation terminating on them. (b) An example of a symmetric excitation. Two $e$-excitations live on the boundary of a bulk dual string excitation, depicted in red.}
	\label{figBoundaryLogicals}
\end{figure}

Now we consider excitations that respect the symmetry. On the boundary, we see that boundary excitations are symmetric only if they are accompanied by a bulk string excitation. In particular, a string operator $Z(l)$ creating $e$ particles at vertices $\mu$ and $\nu$ is made symmetric by attaching a bulk string operator $Z(E')$ whose boundary is at the location of the two particles $\partial E' = \{ \mu, \nu\}$ (i.e., $l\cup E'$ is a cycle). Similarly, the dual string operator $\tilde{X}(l')$ that creates $m$ excitations at $\mu'$ and $\nu'$ can be made symmetric by attaching a bulk string operator $Z(F')$ such that the union $l'^{\perp} \cup F'$ is a dual cycle (i.e., has no boundary on the dual lattice). Such excitations will flip cluster stabilizers in the bulk, for all terms $K_e$ with $e \in E'$ and $K_f$ with $f\in F'$, but will only create a pair of $e$ or $m$ particles on the boundary at their endpoint.

The following two lemmas characterise the valid configurations of excitations in the presence of symmetry. 

\begin{lem}\label{lemPDCond}
	The toric code boundary is both primal-condensing and dual-condensing.
\end{lem}

\begin{proof}
We first show that it is dual-condensing. We can decompose any cycle $l\subset E$ into two components: $l = l_{\text{int}} \cup l_{\text{boun}}$ where $l_{\text{int}} = l \cap \interior{E}$ is its interior component and $l_{\text{boun}} = l \cap \partial E$ is its boundary component. As we have seen, $Z(l_{\text{int}})$ anti-commutes with all terms $K_e$ with $e\in l_{\text{int}}$ and commutes with all other terms. Also, $Z(l_{\text{boun}})$ commutes with all terms apart from $\overline{A}_v$ with $v\in \partial l_{\text{boun}}$. Therefore any bulk dual loop excitation given by an operator $Z(l)$ may be translated to a boundary using a series of local symmetric moves (translations may be performed by sequentially applying $Z(c)$ operations for some small cycle $c$). The dual loop excitation can then be piecewise absorbed upon contact with the boundary. 

To show primal-condensing, the argument is similar. We decompose any dual-cycle $l' \subset F$ into two components $l' = l_{\text{int}}' \cup l_{\text{boun}}'$ where $l_{\text{boun}}' = l' \cap F_{\text{boun}}$ and $l_{\text{int}}' = l' \cap (F \setminus F_{\text{boun}})$ where $F_{\text{boun}} = \{f \in F ~|~ \partial f\cap \partial E \neq \emptyset\}$. Intuitively, $F_{\text{boun}}$ is the set of faces that contain one edge on the boundary of the lattice. Then $Z(l_{\text{int}}')$ anti-commutes with all terms $K_f$ with $f\in l_{\text{int}}'$ and commutes with all other terms. Now find a string $t\subset \partial E$ on the boundary such that $t^{\perp} = l_{\text{boun}}'$ (recall $t^{\perp} = \{f \in \interior{F}: \partial f \cap t \neq \emptyset \}$). Such a string can always be found. Now $Z(l_{\text{boun}}')$ itself doesn't commute with all bulk cluster terms $K_f$, but $Z(l_{\text{boun}}') X(t) = \tilde{X}(t)$ is a dressed string operator that commutes with all terms apart from the plaquettes $\overline{B}_f$ with $f\in \delta{t}$. Then similarly to the previous case, any primal loop excitation in the bulk can be translated to the boundary where it can be piecewise absorbed by sequentially applying local Pauli $X$ operators.
\end{proof}

As we have seen, primal and dual excitations need only be closed loops modulo the toric code boundary, where they can terminate as an anyonic $m$ or $e$-type excitations, respectively. The following lemma states that in fact these anyonic excitations can only exist if they are at the end of a bulk string excitation. 

In the following, for any subset of faces $f$, let $\delta f \subset Q$ be the set of volumes that each contain an odd number of faces of $f$ on their boundary ($\delta f = \{ q \in Q ~:~ |\partial q \cap f | \neq 0 \mod 2\}$). 
\begin{lem}\label{lemAnyonCoupling}
	In the 1-form symmetric sector, $e$-excitations can be located at sites $V_e\subset \partial V$ if and only if accompanied by a dual bulk string excitation supported on $l\subset E$ satisfying $\partial l = V_e$. Similarly, $m$-excitations can be located at sites $F_m \subset \partial F$ if and only if accompanied by a primal bulk string excitation supported on $l' \subset F$ satisfying $\delta F_m = \delta l'$.
\end{lem}

\begin{proof}
	For the $e$-excitations, we have the following constraint: For every vertex operator $\overline{A}_v$, $v\in \partial V$, there exists a unique $e\in \interior{E}$ such that $S_v = \overline{A}_v K_e$ (can be seen upon inspection of Fig.~\ref{figBoundarySyms}). As $S_v = +1$ in the ground space, it must also be for any excitations produced by a symmetric process. Therefore any flipped term $\overline{A}_v$ must be accompanied by a uniquely determined flipped bulk term $K_e$. As every dual qubit on an edge $e$ is in the support of two symmetry generators $S_{v_1}$ and $S_{v_2}$, which also must be preserved, the flipped term $K_e$ must be part of a string excitation can only terminate at another flipped term  $\overline{A}_w$, $w\in \partial V$.
	
	For the $m$-excitations, the argument is the same after noting the following constrain between bulk and boundary excitations: For every plaquette operator $\overline{B}_f$, $f\in \partial F$, there exists a unique $q\in Q$ such that $S_q = \overline{B}_f \prod_{f'\in \partial q} K_{f'}$.
\end{proof}

\subsubsection{Energetics of boundary excitations}\label{secBoundaryExcitationDist}
For any two vertices $v,v'\in V$ let $d(v,v')$ denote the lattice distance between $v$ and $v'$ as $d(v,v') = \min_{l\subset E}\{|l| ~:~ \partial l = (v,v')\}$. Namely, it is the smallest number of edges required to connect the two vertices. Similarly, for any two faces $f,f' \in F$, $d(f,f')$ is defined to be the lattice distance between $f, f'$ on the dual lattice (where 3-cells are replaced by vertices, faces by edges, edges by faces, and vertices by 3-cells). Also, recall $\Delta_{\text{gap}} =2$ is the energy gap. 
\begin{lem}\label{lemEnergetics}
	For the model $H$ defined on the half Euclidean (3D) space, the minimal energy cost to symmetrically create a pair of $e$-excitations ($m$-excitations) at positions $x,x'$ is given by $(d(x,x') + 4)\Delta_{\text{gap}}$.
\end{lem}
\begin{proof}
	Consider the process of creating a pair $e$-excitations on the boundary at positions $x(i_0), x(i_0)'$ and then moving them to positions $x=x(i_k)$, $x'=x(i_k)'$ using a sequence of moves labelled by $i_1,\ldots i_k$. The positions of the excitations at steps $i_j$ are given by $x(i_j), x(i_j)'$. From Lemma~\ref{lemAnyonCoupling} at every step $i_j$, the excitations must be accompanied by a dual string excitation in the bulk supported on $l(i_j)\subset \interior{E}$ with $\partial l(i_j) = (x(i_j), x(i_j)')$. The energy cost of the string $l(i_j)$ is given by its length $|l(i_j)|\Delta_{\text{gap}}$ which minimally is $(d(x(i_j), x(i_j)')+2)\Delta_{\text{gap}}$. Adding in the energy cost $2\Delta_{\text{gap}}$ of the two $e$-excitations gives the result. The $m$-excitations follow analogously. 
\end{proof}
We will use this lemma in the following subsections to derive the symmetric energy barrier. 

\subsubsection{Comparison: Trivial bulk and toric code boundary}\label{secTrivialBulkBoundary}
We contrast this with respect to the trivial model. Namely, consider the trivial model $H^{(0)} = H_{\interior{\calL}}^{(0)} + H_{\partial \calL}^{(0)}$, with $H_{\interior{\calL}}^{(0)}$ the trivial paramagnet defined in Eq.~(\ref{eqHamTriv}) and 
\begin{equation}
H_{\partial \calL}^{(0)} = - \sum_{v \in \partial \calL} {A}_v -\sum_{f \in \partial F} {B}_f.
\end{equation}
with $A_v$ and $B_f$ the undressed toric code terms of Eq.~(\ref{eqToricTriv}). The trivial model $H^{(0)}$ can be connected to our model $H$ using the non-symmetric circuit of Eq.~(\ref{eqRBHcircuit}). Lemma~\ref{lemPDCond} still holds for the trivial model, however Lemma~\ref{lemAnyonCoupling} and subsequently Lemma~\ref{lemEnergetics} do not. Indeed one can symmetrically create a pair of flipped plaquettes $B_f$ using a string of $X$ operators, without creating any bulk excitation. The coupling between boundary anyons and bulk strings is crucial for self-correction, as otherwise the anyons remain deconfined on the boundary. We discuss how this conditions results from the anomalous SET order of the boundary, and the SPT order of the bulk in Sec.~\ref{secBulkBoundary}.

We now have symmetry and spectral properties of the toric code boundary. In Sec.~\ref{sec:cubicRBH} we will discuss the ground space degeneracy of this boundary in the context of the full model and its associated lattice topology. 

\subsection{Boundaries III: Other types of boundaries}\label{secOtherBoundaries}

The toric code boundary is not the only choice of symmetric boundary condition, and in our construction of a SCQM we will make use of several other boundaries. In particular, our construction requires the existence of boundaries with certain condensation properties to ensure that all logical operators of the code can be implemented using sequences of local symmetric moves. While a rigorous classification of the possible boundary theories remains an interesting open problem, the toric code boundary conditions of the previous section along with the three presented in this section are sufficient for our purposes.

We present these boundaries with a canonical choice of lattice geometry, however we emphasise that they should be viewed as stable topological objects and their important properties are not dependent on microscopic details.  (This can be viewed in analogy to the rough and smooth boundaries of the 2D surface code~\cite{bravyi1998quantum}.)  In particular, if we were to erase a disk within one of these boundaries (by removing the Hamiltonian terms and adding or removing qubits), using the general symmetry-based construction in Sec.~\ref{SecSymBoundary}, we can argue that any possible Hamiltonians that can fill in the disk are determined by the surrounding boundary Hamiltonian and in particular must belong to the same phase. The argument is presented in Appendix~\ref{UVInsensitivity}. Such an argument demonstrates that the particular choice of Hamiltonians in this section may be regarded as canonical representatives of the boundary theories. We leave a rigorous proof of this statement, along with an exhaustive proof of the possible boundary theories for this 1-form SPT phase to future work.

We consider three different boundary geometries that support the following types of gapped boundary Hamiltonians to be used in the code construction, based on their condensation properties as defined in Definition 1 of Sec.~\ref{secBoundarySec}:
\begin{enumerate}
\item the primal boundary Hamiltonian $H_{P}$, which is primal-condensing but not dual-condensing; 
\item the dual boundary $H_{D}$, which is dual-condensing but not primal-condensing; and 
\item the ``sink'' boundary $H_{\text{sink}}$, which is primal-condensing and dual-condensing.
\end{enumerate}
The different boundaries are distinguished by what excitations can condense on them in a symmetric way and they can be thought of in analogy to the $Z$-type and $X$-type boundaries of the surface code (often called rough and smooth)~\cite{bravyi1998quantum}; the primal boundary is chosen to allow primal string-like excitations (i.e., excitations generated by $Z$-strings on primal qubits) to condense, the dual boundary is chosen to allow dual string-like excitations (i.e., excitations generated by $Z$-strings on dual qubits) to condense, and both strings can condense on the sink boundary. There exist nondegenerate, symmetric Hamiltonians consisting of commuting Pauli terms with these properties, as we now show. 

All of the Hamiltonians in this subsection are given by a sum over (potentially truncated) cluster terms
\begin{equation}\label{eqHamP}
H_{\text{boundary}} = -\sum_{f \in \partial{F}} K_{f} -\sum_{e \in \partial{E}} K_{e},
\end{equation}
where $K_e$ and $K_f$ of the form of Eq.~(\ref{eqClusterStateStabilizers}). Similarly, all symmetry generators are determined by Eq.~(\ref{eqRBH1formSymOps}), they may take the form of Fig.~\ref{figCluster1form} or Fig.~\ref{figBoundarySyms} depending on the boundary lattice geometry. The Hamiltonians are all non-degenerate as they are locally equivalent to a trivial paramagnet. We emphasise that although we utilise microscopics details of these boundaries in the analysis of this section, what is important is their topological properties, as encapsulated by lemmas \ref{lemprimcondensing} - \ref{lemSinkBoundary} (which are independent of precise lattice details).

We note that similarly to the bulk case, excitations on the boundary are given by operators $Z(E',F')$ of Eq.~(\ref{eqBulkEx}), for $E'\subset \partial E$ and $F'\subset \partial F$. Such an operator flips precisely the terms $K_e$ and $K_f$ with $e\in E'$ and $f\in F'$, this can be verified by local unitary equivalence with the trivial paramagnet using Eq.~(\ref{eqRBHcircuit}). We note the usual product relation between cluster terms and symmetry operators 
\begin{align}
S_q = \prod_{f \in \partial q} K_f \quad \forall q\in Q, \\
S_v = \prod_{e: v\subset e} K_e \quad \forall v\in V
\end{align}
puts nontrivial constraints on the relationship between bulk and boundary excitations, that we will now explore. 

\subsubsection{Primal boundary}
For the primal boundary, we consider `smooth' boundary conditions (in analogy to the smooth or $Z$-type boundary of the surface code). On the boundary, qubits are placed on both boundary edges, and boundary faces, as depicted in Fig.~\ref{figBoundaryPrimal}. On this boundary, the 1-form symmetry generators reduce to the 6-body primal symmetry operators $S_q$ of Fig.~\ref{figCluster1form}, and 5-body dual symmetry operators $S_v$ of Fig.~\ref{figBoundarySyms}. The Hamiltonian terms of $H_D$ are 4-body $K_e$ operators and 5-body $K_f$ operators, as depicted in Fig.~\ref{figBoundaryPrimal}. These terms all commute with the symmetry. 

\begin{lem}\label{lemprimcondensing}
	The primal boundary $H_{P}$ is primal-condensing and not dual-condensing.
\end{lem}
\begin{proof}
	We first show that the boundary is primal-condensing by showing that primal excitations can terminate on it. Firstly, for any pair of faces $f, f' \in \partial F$ on the boundary, any subset of faces $l' \subset \interior{F}$ with $\delta l' = \delta(f\cup f')$ defines a symmetric excitation operator $Z(l')$ (i.e. $[Z(l'), S_q] = 0$ $\forall q\in Q$). This is due to the fact that every boundary face $f$ belongs to a unique 3-cell $q$, meaning each boundary primal qubit is in the support of a unique symmetry generator $S_q$ (as opposed to two in the bulk). As $Z(l')$ flips precisely the terms $K_f$ with $f\in l'$ and commutes with all others, we can locally and symmetrically absorb primal loop excitations near the primal boundary. 
	
	To show that the primal boundary is not dual-condensing, we note that for every dual qubit on some boundary edge $e\in \partial $ is in the support of of two symmetry generators $S_v$, $S_{v'}$. Therefore the only operators $Z(l), l\subset E$ that commute with the 1-form symmetry operators satisfy $\partial l = \emptyset$. This means that dual excitations must form closed loops, even on the boundary. 
\end{proof}

\begin{figure}
	\subfloat[]{	\includegraphics[width=0.44\linewidth]{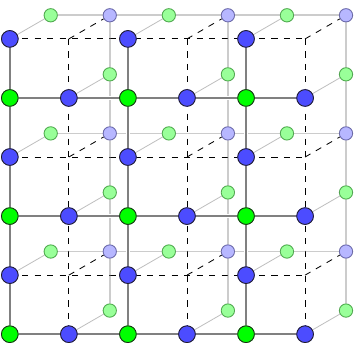}}
	\quad
	\subfloat[]{	\includegraphics[width=0.44\linewidth]{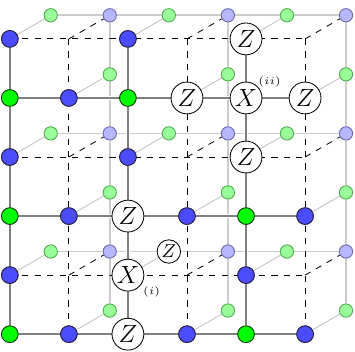}}
	\caption{(a) The lattice at the primal boundary. Primal qubits are depicted in green, while dual qubits are depicted in blue. (b) The primal boundary Hamiltonian $H_{P}$ consists of cluster terms, as depicted by $(i)$ and $(ii)$. Bold lines indicate nearest neighbour relations between qubits, while dashed lines indicate edges of the ambient cubic lattice. }
	\label{figBoundaryPrimal}
\end{figure}

\subsubsection{Dual boundary}

The dual boundary is similar to the primal boundary; it can be obtained by reversing the role of primal and dual qubits on the boundary. In particular, we consider the `rough' boundary conditions depicted in Fig.~\ref{figBoundaryDual} (in analogy with the $X$-type boundary of the surface code). For this boundary, the 1-form symmetry generators reduce to the 6-body dual symmetry operators $S_v$ of Fig.~\ref{figCluster1form}, and 5-body primal symmetry operators $S_q$ of Fig.~\ref{figBoundarySyms}. The Hamiltonian terms of $H_{D}$ are 5-body $K_e$ operators and 4-body $K_f$ operators, as depicted in Fig.~\ref{figBoundaryDual}. These terms all commute with the symmetry.

\begin{lem}
	The dual boundary $H_{D}$ is dual-condensing and not primal-condensing.
\end{lem}
\begin{proof}
	The proof is the same as Lemma~\ref{lemprimcondensing}, exchanging the role of primal and dual qubits. 
\end{proof}

\begin{figure}[h]
	\subfloat[]{	\includegraphics[width=0.45\linewidth]{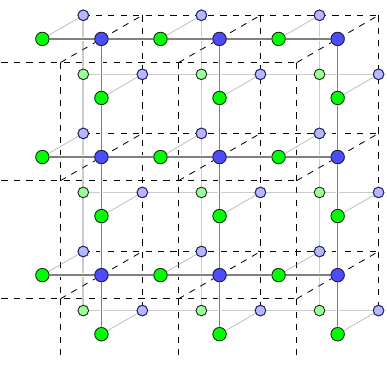}}
	\quad
	\subfloat[]{	\includegraphics[width=0.45\linewidth]{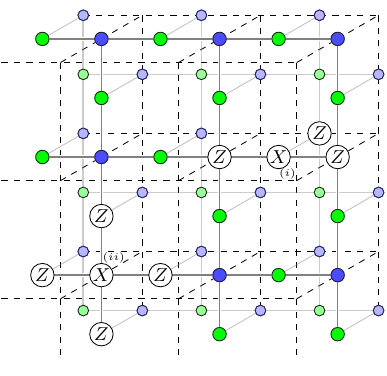}}
	\caption{(a) The lattice at the dual boundary. Primal qubits are depicted in green, while dual qubits are depicted in blue. (b) The dual boundary Hamiltonian $H_{D}$ consists of cluster terms, as depicted by $(i)$ and $(ii)$. Bold lines indicate nearest neighbour relations between qubits, while dashed lines indicate edges of the ambient cubic lattice.}
	\label{figBoundaryDual}
\end{figure}

\subsubsection{Sink boundary}
Finally, we consider the sink boundary. This lattice boundary is again given by the `smooth' boundary conditions of the previous subsection. On the boundary, qubits are placed only on boundary faces, and  not boundary edges, as depicted in Fig.~\ref{figBoundarySink}. On this boundary, both primal and dual 1-form symmetry generators reduce to the 6-body operators of Fig.~\ref{figCluster1form}. The Hamiltonian terms of $H_{PD}$ are 5-body $K_e$ operators and 1-body or 4-body $K_f$ operators, as depicted in Fig.~\ref{figBoundarySink}. These terms all commute with the symmetry.

\begin{lem}\label{lemSinkBoundary}
	The sink boundary $H_{PD}$ is both primal-condensing and dual-condensing.
\end{lem}
\begin{proof}
	The proof is similar to the first part of Lemma~\ref{lemprimcondensing}: we observe that the boundary contains both primal and dual qubits that belong to unique 6-body symmetry generators $S_q$ and $S_v$, respectively (as opposed to two). As such primal and dual excitation chains can symmetrically terminate on these qubits. 
\end{proof}

Finally, we note that gapped interfaces exist between all of these boundaries.  We will demonstrate this fact explicitly in the following subsection.

\begin{figure}
	\subfloat[]{	\includegraphics[width=0.42\linewidth]{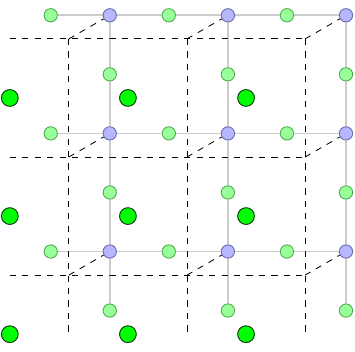}}
	\quad
	\subfloat[]{	\includegraphics[width=0.42\linewidth]{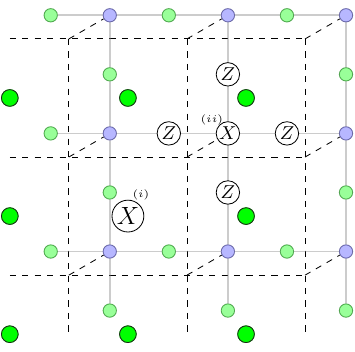}}
	\caption{(a) The lattice at the sink boundary. Primal qubits are depicted in green, while dual qubits are depicted in blue. (b) The dual boundary Hamiltonian $H_{PD}$ consists of cluster terms, as depicted by $(i)$ and $(ii)$. The primal qubits on the boundary surface have no neighbours (meaning the corresponding cluster term is given simply by Pauli $X$). Bold lines indicate nearest neighbour relations between qubits, while dashed lines indicate edges of the ambient cubic lattice.}
	\label{figBoundarySink}
\end{figure}

\subsection{The cubic RBH code}\label{sec:cubicRBH}
We now use these various boundaries to construct a the code that is self correcting under 1-form symmetries, we call the model the cubic RBH model. 

\subsubsection{The lattice}
The lattice $\calL$ we consider has the topology of a 3-ball. Namely, we consider cubic boundary conditions: the lattice is a cubic lattice with dimensions $d \times d \times d$, with six boundary facets, depicted in Fig.~\ref{figBoundaryRBHSix}. The bulk of the model is given by the usual RBH cluster Hamiltonian, while on each of the six boundary facets we choose one of four different boundary conditions. Namely, one of the six boundary faces is chosen to support the logical information using a dressed toric code $H_{\partial \calL}$ -- which we will call the toric code boundary -- and the remaining five boundary faces supports either a primal boundary, a dual boundary or a sink boundary, as depicted in Fig.~\ref{figBoundaryRBHSix}.

The lattice must terminate on each of these boundary facets according to the boundary conditions outlined in the previous two subsections. In Fig.~\ref{figBoundaryRBHSix} we show a small example of the lattice when viewed from the direction of the toric code (i.e., $H_{\partial \calL}$) boundary. Note in particular that the toric code boundary facet has planar boundary conditions due to the way it terminates on the primal and dual boundaries. Namely, the top and bottom edges of the toric code boundary facet are known as rough edges, and the left and right edges are known as smooth edges.

\begin{figure}[h]
	\centering
	\subfloat[]{	\includegraphics[width=0.455\linewidth]{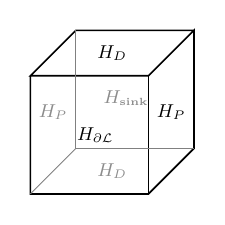}}
	\quad
	\subfloat[]{	\includegraphics[width=0.45\linewidth]{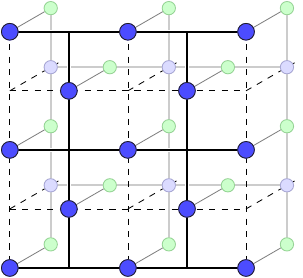}}
	\caption{(a) The boundaries of the cubic RBH model. $H_{\partial \calL}$ is the toric code boundary, $H_P$ and $H_D$ are the primal and dual boundaries respectively, and $H_{\text{sink}}$ is the sink boundary. (b) The lattice for the toric code boundary $H_{\partial \calL}$. The top and bottom edges are called rough boundary conditions while the left and right edges are called smooth boundary conditions. The Hamiltonian consists of the negative sum of all star and plaquette terms $\overline{A}_v$, $\overline{B}_f$ from Eq.~(\ref{eqDressedTC}). Dashed lines denote edges of the cubic lattice.}
	\label{figBoundaryRBHSix}
\end{figure}

\subsubsection{The Hamiltonian}
The Hamiltonian decomposes into bulk and boundary components. The bulk Hamiltonian is given by the usual RBH cluster Hamiltonian $H_{\interior{\calL}}$ of Eq.~(\ref{eqBulkHam}). The boundary Hamiltonians come in four different types, firstly, on the toric code boundary we put the dressed toric code Hamiltonian $H_{\partial \calL}$ of Eq.~(\ref{eqDressedToricCodeBound}). Dressed toric code terms are truncated near the rough and smooth edges. In particular, the plaquette terms $\overline{B}_f$ are truncated near the rough boundaries, while the star terms $\overline{A}_v$ are truncated near the smooth boundaries. The Hamiltonians $H_{P}$, $H_{D}$, and $H_{\text{sink}}$ on the primal, dual and sink boundaries, can all be expressed in the form $H_{\text{boundary}}$ of Eq.~(\ref{eqHamP}). Terms in these Hamiltonians are cluster terms that are potentially truncated, depending on what boundary they reside on. 

We note that all of these boundaries meet at gapped interfaces. In particular, the lattice structure at the edge lines and corners is explicitly depicted in the Fig.~\ref{figBoundaryRBHSix}. The symmetry operators are again generated by $S_v$ and $S_q$ of Eq.~(\ref{eqRBH1formSymOps}). They are 5-body or 6-body operators, depending on if they are near a particular boundary. All Hamiltonian terms are symmetric and mutually commuting.

\subsubsection{The ground space}
As discussed, the bulk Hamiltonians $H_{\interior{\calL}}$, along with the boundary Hamiltonians $H_{P}$, $H_{D}$, and $H_{\text{sink}}$ are all non-degenerate. The overall degeneracy manifests on the toric code boundary $H_{\partial \calL}$. In particular, for the planar boundary conditions on the toric code boundary, there is a 2-fold degeneracy. This can be easily verified by its local unitary equivalence with the planar code, which encodes one logical qubit.

\subsubsection{Logical operators and codespace}

The toric code Hamiltonian $H_{\partial \calL}$ encodes one logical qubit, with string logical operators $\overline{X}$ and $\overline{Z}$ running between opposite pairs of edges of the boundary face. In particular, the logical operators are given by 
\begin{equation}\label{eqPlanarLogicals}
\overline{X} = \prod_{e \in a_d} X_e \prod_{f \in a_d^{\perp}}Z_f, \qquad \overline{Z} = \prod_{e \in b_d} Z_e,
\end{equation}
where  $a_d$ is a dual-cycle on the boundary (meaning it it a cycle on the dual of the boundary lattice) that runs between the two smooth edges, $b_p$ is a cycle on the boundary that runs between the two rough edges, and $a_d^{\perp} = \{f \in \interior{F}: \partial f \cap a_d \neq \emptyset \}$. These logicals are depicted in Fig.~\ref{figPlanarLogicals}. Note in particular, that such strings are symmetric, as the top and bottom boundary facets are dual-condensing, while the left and right are primal-condensing. 

\begin{figure}[h]
	\centering
	\subfloat[]{	\includegraphics[width=0.45\linewidth]{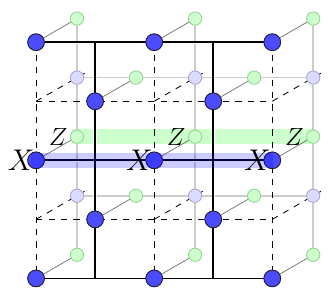}}
	\quad
	\subfloat[]{	\includegraphics[width=0.45\linewidth]{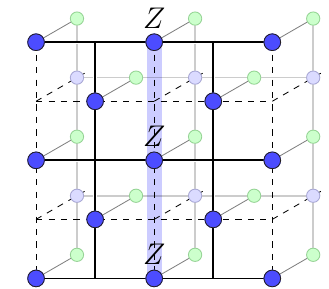}}
	\caption{Logical operators for the toric code boundary $H_{\partial \calL}$. (a) Logical $\overline{X}$ runs between the left and right smooth edges. (b) Logical $\overline{Z}$ runs between the top and bottom rough edges. Dashed lines denote edges of the cubic lattice.}
	\label{figPlanarLogicals}
\end{figure}

\subsubsection{Logical operator decomposition}
In this model, logical operators admit symmetric local decompositions, as we now demonstrate. The toric code Hamiltonian encodes one logical qubit, with string logical operators $\overline{X}$ and $\overline{Z}$ running between opposite pairs of edges of the boundary face. These logicals are given by Eq.~(\ref{eqPlanarLogicals}). In order to implement either logical operators ($\overline{X}$ or $\overline{Z}$) through a sequence of local moves, we will also create a large bulk excitation. (Note this is expected, as we claim the model is self-correcting, we must necessarily traverse a large energy barrier to implement a logical operator). This large bulk excitation can then be absorbed by the sink boundary in order to return to the codespace. Importantly, $e$-excitations ($m$-excitations) can be symmetrically created and destroyed at the rough edge (smooth edge) of the toric code boundary. In fact, implementing a logical $\overline{Z}$ ($\overline{X}$) operator can be viewed as a process creating an $e$-excitation ($m$-excitation) from one rough (smooth) edge to the opposite rough (smooth) edge. The strategy is outlined in Fig.~\ref{figBoundaryRBHSixError}.

\begin{lem}\label{lemSymDecomp}
	Both logical $\overline{X}$ and $\overline{Z}$ of the cubic RBH model admit symmetric local decompositions. 
\end{lem}
\begin{proof}
We first consider a symmetric local decomposition of $\overline{Z}$. Consider a string operator $Z(c)$, $c\subset E$ supported on the dual qubits near the code boundary, as in Fig.~\ref{figBoundaryRBHSixError}. Grow this string operator until we achieve $Z(l+ l') \equiv Z(l) Z(l')$, where $l+l'$ is a contractible loop (and therefore achievable by local symmetric moves), $l$ is a string running between the top and bottom rough edges, and $l'$ is a string in the bulk with the same boundaries as $l$. Thus $Z(l)$ is a logical $\overline{Z}$ operator, and $Z(l')$ is an operator causing a bulk string-like excitation, anchored between the two dual boundaries. We then consider translating the bulk excitation caused by $Z(l')$ to the sink boundary, following Fig.~\ref{figBoundaryRBHSixError} (which can be achieved with local symmetric moves as the two loops are homologous). This operator, and the corresponding excitations, can then be absorbed by the sink boundary as it is dual-condensing.

Logical $\overline{X}$ operators can be decomposed in a similar way. First, consider the same process as above to produce a string operator $Z(l')$, $l\subset F$ supported on the primal qubits anchored between the opposite primal boundaries (can be achieved in the same way, as the sink boundary is primal-condensing). $Z(l')$ can be translated adjacent to the code boundary, such that $l' = a_d^{\perp}$ for some dual-cycle on the boundary $a_d$. One can then apply a sequence of Pauli-$X$ operators along $a_d$, giving logical $\overline{X}$.
\end{proof}

\begin{figure}[h]
	\includegraphics[width=0.99\linewidth]{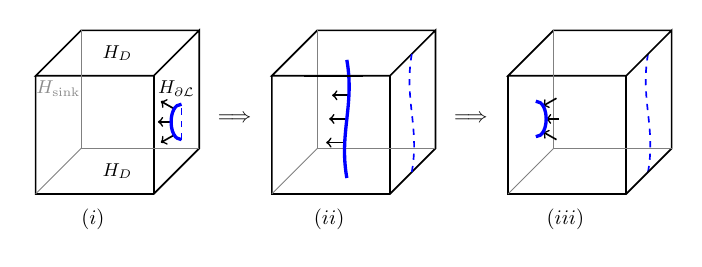}
	\caption{Implementing a logical $\overline{Z}$ operator through a sequence of local moves. $(i)$ An error chain $Z(c)$ supported on dual qubits (the union of the solid blue and dashed blue lines) is created near the toric code boundary. This error chain creates string excitations in the bulk (solid blue), and anyonic excitations where it meets the toric code boundary. $(ii)$  The loop is grown until it consists of a logical operator $\overline{Z}$ (dashed blue line) along with a large bulk excitation (solid blue), anchored between the two dual boundaries. $(iii)$ The bulk excitation is moved to the sink boundary, where it can be absorbed. The whole process results in a logical $\overline{Z}$. Logical $\overline{X}$ operators can be implemented in a similar way, where an error loop on the primal lattice is grown and propagated, and an additional chain of Pauli $X$ errors is also propagated along the toric code boundary.}
	\label{figBoundaryRBHSixError}
\end{figure}

\subsubsection{The energy barrier}
As we have seen, when the dynamics are restricted to the 1-form symmetric sector, bulk excitations form collections of closed loop-like objects. Secondly, boundary anyonic excitations only appear at the end of a bulk string-like excitation. This coupling of the thermal properties between bulk and boundary in the presence of symmetry, is enough to achieve a diverging symmetric energy barrier (as defined in Eq.~(\ref{eqSymEnergyBarrier})).

\begin{defn}\label{defndist}
	We define the lattice width $d$ of the cubic RBH model as $d = \min \{d_{Z} , d_{X} , d_{\text{cond}} \}$, where $d_{Z}$ is the smallest lattice distance between the two rough edges of the toric code boundary, $d_{X}$ is the smallest lattice distance between the two smooth edges of the toric code boundary, and $d_{\text{cond}}$ is the smallest lattice distance between the toric code boundary and the sink boundary. 
\end{defn}
Note that $\min\{d_{Z} , d_{X} \}$ is the usual (code) distance of the planar code on the same boundary. For any edge $e\in \partial E$ (face $f\in \partial F$) we define $d_{\text{cond}}(e)$ ($d_{\text{cond}}(f)$) as the lattice distance to the nearest dual-condensing (primal-condensing) boundary. Recall also the lattice distance $d(x,x')$ defined in Sec.~\ref{secBoundaryExcitationDist}.

\begin{lem}\label{lemSinkEnergy}
	Let $C \subset \partial E \cup \partial F$ denote the positions of a general configuration of boundary anyons. Then the energy cost to symmetrically create this configuration is lower bounded by $( \tilde{d}_C + |C|)\Delta_{\text{gap}}$, where 
	\begin{equation}\label{eqDistConfig}
	\tilde{d}_C = \min_{P\in \hat{P}} \left\{\sum_{\{a\}, \{b,c\} \in P} d_{\text{cond}}(a) + d(b,c) \right\}
	\end{equation}
	where $P$ is a partition of the elements of $C$ into pairs $\{b,c\}$ of the same type or singletons $\{a\}$, and $\hat{P}$ is the set of all such partitions.
\end{lem}
\begin{proof}
	This is the generalisation of Lemma~\ref{lemEnergetics} to the cubic RBH model. The proof follows in the same way, where we additionally note that each $e$ ($m$) anyon may be connected by a bulk loop excitation to either another $e$ ($m$) anyon, or to an appropriate dual-condensing (primal-condensing) boundary. As such, the smallest energy cost is obtained by finding the total length of the (shortest) perfect match for all anyons, where anyons are allowed to pair with their respective boundary. The energy cost is then obtained by scaling the length of the excitations by the gap $\Delta_{\text{gap}}$, and adding in the contribution for each anyon. 
\end{proof}

\begin{thm}
	The symmetric energy barrier for a logical fault in the cubic RBH model is lower bounded by 
	\begin{equation}
	{d\cdot\frac{\Delta_{\text{gap}}}{2} - r'},
	\end{equation}
	where $d$ is the lattice width, defined in Def.~\ref{defndist}, and $r'$ is constant (independent of lattice size). 
\end{thm}
\begin{proof}
Let $\{l^{(k)} ~|~ k = 1, \ldots N\}$ be any sequence of operators such that each $l^{(k)}$ is symmetric, $l^{(k)}$ and $l^{(k+1)}$ differ only locally, $l^{(1)} = I$ and $l^{(N)}$ is a logical operator supported on either $a_d$ dual-cycle or the $b_p$ cycle of Eq.~(\ref{eqPlanarLogicals}). Let $r$ be the largest range of any operator $l^{(k)}l^{(k+1)}$ for any $k\in 1,\ldots N$, which is assumed to be constant. 

By locality of $l^{(k)}$, we must traverse an intermediate state that has a nonzero number of anyonic excitations on the code boundary. Moreover, since at each time step the separation between anyons can only change by a constant amount, to achieve a nontrivial logical operator, there is a time step $k' \in \{1,\ldots N\}$ with a configuration of anyons given by $C_{k'}$, such that $d_{C_{k'}} \geq \min\{\lfloor d_X/2 \rfloor, \lfloor d_Z/2 \rfloor\} -r$. Here, ${d}_{C_{k'}}$ is given by the (minimum length) perfect match of all anyons on $\partial L$, where anyons can be matched with the boundaries they can condense on. Note that  $\tilde{d}_{C_{k'}} \geq \min\{d_{C_{k'}}, d_{\text{cond}}\}$, where $\tilde{d}_{C_{k'}}$ is defined in Eq.~(\ref{eqDistConfig}). Then by Lemma~\ref{lemSinkEnergy}, we have that the energy cost of the configuration $C_{k'}$ is at least $(\tilde{d}_{C_{k'}} + |C_{k'}|)\Delta_{\text{gap}}$ which is lower bounded by $(\min\{d_{C_{k'}}, d_{\text{cond}}\} + |C_{k'}|-r) \Delta_{\text{gap}}$. Using the definition of the lattice width and letting $r' = r \Delta_{\text{gap}}$, the result follows.
\end{proof}

This proof gives a conservative lower bound on the energy barrier, but it is sufficient for our purposes. In particular, as the lattice width $d$ grows with the number of qubits, we have a macroscopic energy barrier. In other words, the energy barrier for a logical fault grows with the size of the system. 

\subsubsection{Self-correction}

We have shown that the 1-form symmetric cubic RBH model inherits a macroscopic energy barrier to a logical fault, due to the string-like nature of excitations resulting from the 1-form symmetry together with its coupling of bulk and boundary excitations. The question is whether this is sufficient for an unbounded memory time. In Appendix~\ref{PeierlsR}, we give an argument following the well-known Peierls argument (see also Ref.~\cite{Dennis}) to show that this energy barrier implies self-correction of the 1-form symmetric RBH model.  In brief, we estimate the probability that an excitation loop $l$ of size $w$ emerges within the Gibbs ensemble at inverse temperature $\beta$. We show that large loop errors are quite rare if the temperature is below a critical temperature $T_c$, and we give a lower bound on $T_c$ at $2/\log(5)$. As such, if the error rate is small enough (that is, the temperature is low enough), then the logical information in the code is stable against logical errors and the encoded information on the boundary will be protected for a time growing exponentially in the system size. 

Along with the memory time, we have therefore met all of the requirements of a symmetry-protected, self-correcting quantum memory. In particular, we have shown that all operators admit a symmetric, local decomposition in Lemma~\ref{lemSymDecomp}. Additionally, the ground space of this system is perturbatively stable, as it meets the TQO stability conditions of Ref.~\cite{TopoStability}. Finally, as a code, it admits an efficient decoder~\cite{Dennis,RBH}. Therefore this model meets the requirements for a self-correcting quantum memory when protected by the $\zz_2^2$ 1-form symmetry.

\subsection{Encoding with more general topologies}
One may ask what other boundary conditions and topologies can be used to construct a self-correcting code under 1-form symmetries. In this subsection we outline one other choice and rule out a number of others. 
 
In particular, note that in the previous discussion we could replace the sink Hamiltonian with another toric code Hamiltonian, as it is both primal and dual condensing. While the degeneracy of the ground space increases by another factor of two in this case, we do not get an increase in the number of qubits that we can encode. This is because the two opposite toric code boundaries must always be correlated as dictated by the symmetry: labelling the two codes as $L$ and $R$, there is no local symmetric decomposition of individual logical operators $\overline{Z}_L$ and $\overline{Z}_R$ ($\overline{X}_L$ and $\overline{X}_R$), but only of the product $\overline{Z}_L \otimes \overline{Z}_R$ ($\overline{X}_L\otimes \overline{X}_R$). This property is similar to theory of SPT phases in one dimension, where the two separate degenerate boundary modes of a 1D chain cannot be independently accessed in the presence of symmetry. 

Similarly, one could remove the primal and dual boundaries, by considering the lattice $\calL$ with a topology of $T^2 \times I$, where $T^2$ is the torus and $I=[0,1]$ is the interval. On each side, $T^2 \times \{0\}$, $T^2 \times \{1\}$ we choose toric code boundary conditions and define a toric code Hamiltonian $H_{\partial \calL}$. With this topology, the ground space of the system is $2^4$-fold degenerate (as each boundary toric code has a degeneracy $d_g = 2^{2g}$ where $g$ is the genus of the 2D manifold it is defined on, with $g=1$ for the torus). For each toric code, one can define logical operators 
\begin{equation}\label{eqToricLogicals}
\overline{X}_1 = \prod_{e \in a_d} X_e \prod_{f \in a_d^{\perp}}Z_f, \qquad \overline{Z}_1 = \prod_{e \in b_p} Z_e,
\end{equation}
and
\begin{equation}
\overline{X}_2 = \prod_{e \in b_d}X_e \prod_{f \in b_d^{\perp}}Z_f, \qquad \overline{Z}_2 = \prod_{e \in a_p} Z_e,
\end{equation}
for cycles $a_p$, $b_p$ and dual-cycles $a_d$, $b_d$ wrapping around the two nontrivial cycles of the torus labelled by $a$ and $b$. Similarly, we can only make use of one of the toric codes, as the two copies are correlated under the 1-form symmetry. In other words, we do not have a symmetric decomposition of all logical operators, only a subgroup of them.

\subsubsection{Topological obstruction to logical decompositions}
The issue of finding choices of boundary conditions that allow for symmetric local decompositions of logical operators is nontrivial. For example, on a solid torus $D^2 \times S^1$, with $D^2$ a disk and $S^1$ a circle (depicted in Fig.~\ref{figSolidTorus}), we cannot encode any logical qubits. Although the boundary of the solid torus is a torus, there does not exist symmetric local decompositions of logical operators supported on the $b$ cycle of Fig.~\ref{figSolidTorus}. For example, logical operators $\overline{Z}$ supported on the $b$ cycle (in Fig.~\ref{figSolidTorus}) cannot be created by a sequence of local, symmetric operators, because any such sequence results in a homologically trivial (contractible) cycle. This phenomenon will always occur for codes that live on the boundary of a 3-manifold due to the following fact: for any 2-manifold, precisely half of the noncontractible cycles (if they exist) become contractible when the manifold is realised as the boundary of a 3-manifold~\cite{hatcher2000notes}. This justifies our consideration of the more involved boundary conditions of the previous subsection. 

\begin{figure}
	\centering	\includegraphics[width=0.4\linewidth,angle=-00]{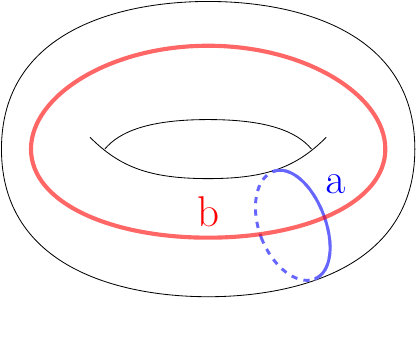}	
	\caption{The solid torus. The boundary of the solid torus is a torus, where two nontrivial cycles $a$ and $b$ are depicted. One might expect to be able to encode two logical qubits in the system, however any operator supported on the $b$ loop does not admit a symmetric local decomposition.}%
	\label{figSolidTorus}
\end{figure}

\subsection{Bulk boundary correspondence at nonzero temperature}\label{secBulkBoundary}

As shown above, the 1-form symmetries constrain the form of the excitations in the model and give rise to an energy barrier, and self-correction.  These 1-form symmetries are a very strong constraint, and one may ask if a code is trivially guaranteed to be self-correcting whenever such symmetries are enforced.  (As a example of a strong symmetry leading trivially to self-correction, consider the toric code where the symmetry of the full stabilizer group is strictly enforced.) 

In this section we show that the 1-form symmetry, although strong, is itself not sufficient to lead to self-correction unless the bulk is SPT ordered (such as in the previous models).  Specifically, we show that self-correction under 1-form symmetries depends on the bulk SPT order of the model, establishing a bulk-boundary correspondence for SPTs at nonzero temperature. Recall, at zero temperature, the correspondence is that a system with nontrivial SPT order in the bulk must have a protected boundary theory -- meaning it is gapless or topologically ordered -- whenever the symmetry is not broken \cite{SPTBdry1, SPTBdry2}. Here we show that the bulk boundary correspondence holds at nonzero temperature in the RBH model; that the stability of the boundary toric code phase (i.e.,~whether or not we have a SCQM) depends on the bulk SPT order at nonzero temperature. 

In order to make this connection, we recall a formulation of phase equivalence due to Chen \textit{et al.}~\cite{WenTopo}. Namely, two systems belong to the same phase if they can be related by a local unitary transformation (a constant depth quantum circuit), up to the addition or removal of ancillas. Importantly, with symmetries $S(g)$ present, the local unitary transformations must commute with the symmetry and the ancillas that are added or removed must be in a symmetric state. 

We now remark on the earlier claim on the necessity of the SPT nontriviality of the bulk to achieve self-correction. To do so, we first note that the symmetric energy barrier is invariant under symmetric local unitaries (that is, it is a phase invariant). Indeed consider two Hamiltonians $H_A$ and $H_B$ (defining quantum memories) in the same phase. Then in particular, we have $H_A +H_{\mathcal{A}}$ and $H_B$ are related by a symmetric local unitary $U$, where let $H_{\mathcal{A}}$ consists of a sum of local projections on the ancillas $\mathcal{A}$ into a symmetric state.  Since $H_A$ and $H_A + H_{\mathcal{A}}$ differ only by a sum of non-interacting terms on the ancilla, they have the same energy barrier. Let $\overline{X}$ be a logical operator for $H_A$, and consider a local decomposition $\{l_{X}^{(k)} ~|~ k = 1, \ldots N\}$ of $\overline{X}$ (recall $l_{X}^{(1)}=I$ and $l_{X}^{(N)} = \overline{X}$, and $l_{X}^{(k)}$ and $l_{X}^{(k+1)}$ differ only by a local operator). This is also a logical decomposition for $H_A + H_{\mathcal{A}}$. Then $\{Ul_{X}^{(k)}U^{\dagger} ~|~ k = 1, \ldots N\}$ constitutes a local decomposition for a logical operator of $H_B$, with the same energy barrier. This works for all choices of logical operators $\overline{X}$ and the models have the same symmetric energy barrier. 

The invariance of the energy barrier requires us to consider a SPT-nontrivial bulk to achieve self-correction in the presence of 1-form symmetries. Indeed, if we instead considered the SPT-trivial model $H_{\interior{\calL}}^{(0)}$ of Eq.~(\ref{eqHamTriv}) with undressed toric code terms of Eq.~(\ref{eqToricTriv}) on the boundary in the presence of 1-form symmetries, we see that there is no energy barrier, in the following way. Consider the logical $\overline{X}$ operator, which is given by a product of Pauli $X$ operators supported on a dual cycle on $\partial \calL$ (it is not dressed, unlike the logical $\overline{X}$ of the RBH model $H$). Then the symmetric energy barrier for this error is a constant $2\Delta_{\text{gap}}$, since the process of creating two $m$ particles and wrapping them around a boundary cycle is symmetric, and only flips two $B_f$ plaquettes at any given time. Therefore the trivial model is not self-correcting, even in the presence of 1-form symmetries. In particular, this also gives a simple argument for why $H$ belongs to a distinct SPT phase to $H_{\interior{\calL}}^{(0)}$. Indeed, the SPT ordering in the bulk is crucial to achieving the bulk-boundary anyon coupling of Lemma~\ref{lemAnyonCoupling}, that leads to a confinement of anyons as in Lemma~\ref{lemEnergetics}.

This bulk boundary correspondence (at nonzero temperature) holds for systems with onsite symmetries too; we have argued in Sec.~\ref{secOnsite} that self-correction was not possible on the 2D boundary of a 3D SPT protected by onsite symmetry. This coincides with with the lack of bulk SPT order at $T \textgreater 0$ when the protecting symmetry is onsite, as shown in Ref.~\cite{ThermalSPT}.

\section{The gauge color code protected by 1-form symmetry}
\label{sec:GCC}

We now turn to a model based on the gauge color code in 3 dimensions as our second example of a symmetry-protected self-correcting quantum memory. The gauge color code~\cite{Bombin15} is an example of a topological subsystem code. In this section we study a commuting Hamiltonian model with a 1-form symmetry based on the gauge color code. This model provides another example of a self-correcting quantum memory protected by a 1-form symmetry. 

We first give a brief overview of the gauge color code before defining the Hamiltonian model we are interested in.

\subsubsection{Subsystem codes}

In addition to logical degrees of freedom, subsystem codes contain redundant `gauge' degrees of freedom in the codespace that are not used to encode information. Whereas stabilizer codes are specified by a stabilizer group $\calS$ that is an abelian subgroup of the Pauli group, a subsystem code is specified by a (not necessarily abelian) subgroup $\calG$ of the Pauli group, known as the gauge group. A stabilizer group $\calS$ for the subsystem code can be defined by choosing any maximal subgroup of the center $\calZ(\calG)$ of the gauge group, such that $-\mathbf{1} \notin \calS$. In other words, $\calS \propto \calZ(\calG)$ (in general there are many choices for $\calS$ obtained by selecting different signs for generating elements). As usual, the codespace $C_{\calS}$ is defined as the mutual $+1$ eigenpsace of all elements of $\calS$.

Information is only encoded into the subsystem of $C_{\calS}$ that is invariant under all gauge operators $g\in \calG$. More precisely, we have $C_{\calS} = \calH_{\text{l}}\otimes \calH_{\text{g}}$, where $\calH_{\text{l}}$ is the state space of logical degrees of freedom (elements of $\calG$ act trivially on this space), and $\calH_{\text{g}}$ is the state space of the gauge degrees of freedom (elements of $\calG$ can act nontrivially on this space). There are two types of Pauli logical operators: bare and dressed. Bare logical operators are elements of $C(\calG)$; the centraliser of the gauge group within the Pauli group, meaning they are Pauli operators that commute with all gauge operators. Dressed logicals are elements of $C(\calS)$; the centraliser of the stabilizer group within the Pauli group (meaning they are Pauli operators that commute with all stabilizer operators). Bare logicals act exclusively on logical degrees of freedom and act trivially on the gauge degrees of freedom, while dressed logicals can act nontrivially on gauge degrees of freedom, too. Both types of logicals are identified up to stabilizers (as stabilizers act trivially on the codespace).

\subsection{The gauge color code lattice}

Gauge color codes are defined on lattices known as 3-colexes \cite{BombinColexes1}. In particular, a 3-colex is the result of gluing together 3-cells (polyhedra) such that each vertex is 4-valent (meaning each vertex belongs to 4 edges) and 4-colorable (meaning each polyhedral 3-cell can be given one of four colors such that neighbouring 3-cells are differently colored). Let these four colors be labelled \textbf{r}, \textbf{b}, \textbf{g}, and \textbf{y} (for red, blue, green, and yellow). 

We note that, similar to the RBH model, the gauge color code must have boundaries in order to possess a nontrivial codespace. For concreteness, we consider the tetrahedral boundary conditions of Ref.~\cite{BombinJump}, but one could also consider more general boundary conditions. In the following, we label the Tetrahedral 3-colex by $\calC_3$, which is a set of vertices, edges, faces and 3-cells. Tetrahedral 3-colexes $\calC_3$ are given by cellulations of the 3-ball, whose boundary consists of four facets, each of which must satisfy a certain coloring requirement. To describe this requirement, we first note that each non-boundary edge can be given a single color label, where the color is determined by that of the two 3-cells that it connects. If an edge terminates on a boundary (meaning precisely one of its vertices belongs to the boundary) then its color is determined by unique bulk 3-cell on its other endpoint. Then the boundary coloring requirement is as follows: for each boundary facet, only edges of one color can terminate on the boundary and this color is unique for each facet. We therefore color each boundary facet by the color of the edges that terminate on it.
	
Similarly, each face $f$ in $\calC_3$ can be labelled by pairs of colors $\textbf{uv} \equiv \textbf{vu}$, inherited from the two neighbouring 3-cells that it belongs to. Namely, each non-boundary face is colored by the complement of the two colors on the 3-cells the face is incident to (e.g., a face belonging to a \textbf{r} and \textbf{b} 3-cell is colored \textbf{gy}). Faces on the boundary are colored by the opposite of the color of the boundary and the color of the unique 3-cell they belong to. As such, the boundary of color \textbf{k} consists of plaquettes of all colors \textbf{uv} such that $\textbf{u}, \textbf{v} \neq \textbf{k}$. We arbitrarily choose one of the boundary facets, the \textbf{b} facet, and call this the outer colex $\calC_{\text{out}}$, which consists of the vertices, edges and plaquettes strictly contained on the boundary. This outer colex is therefore a 2-colex (a trivalent and 3-colorable two-dimensional lattice), and can be used to define a 2-dimensional color code. The remainder of the lattice $\calC_3 \setminus \calC_{\text{out}}$ is called the inner colex.

On the outer colex, each plaquette has one of three possible color pairs $\{\textbf{gy}, \textbf{ry}, \textbf{rg}\}$, which we relabel for simplicity according to $\textbf{gy}\leftrightarrow \textbf{A}$, $\textbf{ry}\leftrightarrow \textbf{B}$, $\textbf{rg}\leftrightarrow \textbf{C}$ as in Fig.~\ref{fig2dCC}. Each edge of the outer colex neighbours two plaquettes of distinct colors, we color each edge the third remaining color. Moreover, each of the three boundaries of the outer colex can be given a single color according to what color edges can terminate on them, as depicted in Fig.~\ref{fig2dCC}. 

\begin{figure}
	\centering
	\subfloat[]{
	\includegraphics[height = 0.42\linewidth, width=0.49\linewidth,angle=-00]{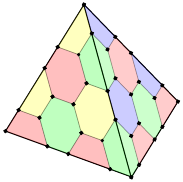}}
	\quad
	\subfloat[]{
	\includegraphics[width=0.44\linewidth]{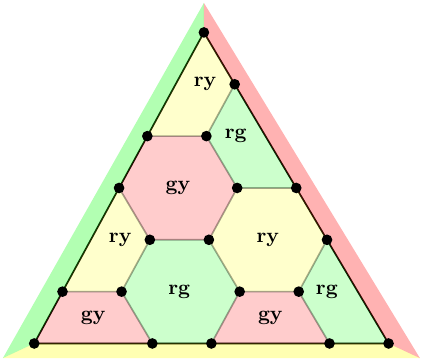}}
	\caption{(a) The tetrahedral 3-colex. (b) The \textbf{b} boundary of the tetrahedral lattice consists of faces that are colored \textbf{uv} with $\textbf{u}, \textbf{v} \neq \textbf{b}$, which are then relabelled according to $\textbf{gy} \leftrightarrow \textbf{A}$, $\textbf{ry} \leftrightarrow \textbf{B}$, and $\textbf{rg} \leftrightarrow \textbf{C}$.\label{fig2dCC}}
\end{figure}

\subsection{The 3D gauge color code}

To each vertex of the lattice $\calC_3$ we place a qubit. The gauge color code is specified by the gauge group $\calG$, which is a subgroup of the Pauli group on $n$ qubits (where $n$ is the number of vertices). The stabilizer group $\calS$ is in the center of the gauge group, consisting of elements of the gauge group that commute with every other element and where the signs are chosen such that $-1 \notin \calS$. For the gauge color code, we have an $X$ and $Z$ gauge generator for each face of the lattice, 
\begin{equation}
\calG= \{ G^X_f, G^Z_f ~|~ f \text{ a face of } \calC_3 \}\,,
\end{equation}
where $G^X_f = \prod_{v \in f} X_v$  and $G^Z_f = \prod_{v \in f} Z_v$ are Pauli operators supported on the face $f$. The stabilizers of the code are given by $X$ and $Z$ on the 3-cells of the lattice 
\begin{equation}\label{eqGCCStabilizer}
\calS = \{ S^X_q , S^Z_q ~|~ q \text{ a 3-cell of } \calC_3 \}\,,
\end{equation}
where $S^X_q = \prod_{v \in q} X_v$  and $S^Z_q = \prod_{v \in q} Z_v$ are Pauli operators supported on 3-cells. Codestates of the gauge color code are the states that are in the $+1$ eigenspace of all elements of the stabilizer group. With the aforementioned boundary conditions, the code encodes one logical qubit, and bare logical operators can be taken to be $\overline{X} = \prod_{v \in \calC_3} X_v$ and $\overline{Z} = \prod_{v \in \calC_3} Z_v$, where the products are over all vertices of the lattice. Importantly, note that equivalent logical operators (i.e., up to products of stabilizers) can be found on the outer colex, namely $\overline{X} \sim \prod_{v \in \calC_{\text{out}}}X_v$ and $\overline{Z} \sim \prod_{v \in \calC_{\text{out}}}Z_v$ are valid representatives. One can find dressed versions of these logicals on the outer colex that are stringlike -- we will discuss this in the following subsection. Similarly to the RBH model, we are therefore justified in viewing the logical information as being encoded on the boundary. 

There are many different Hamiltonians whose ground space contains a representation of the logical degrees of freedom of the gauge color code (here representation means that one can find dressed logicals of the gauge color code that are logical operators for the ground space of a given model). One possible choice of Hamiltonian that represents the GCC logical degrees of freedom in its ground space is given by the sum of all local gauge terms, 
\begin{equation}\label{eqGCCHam}
H_{\calG} = -\sum_f G^X_f - \sum_f G^Z_f,
\end{equation}
which we refer to as the \emph{full GCC Hamiltonian}. This Hamiltonian is frustrated, meaning one cannot exactly satisfy all of the constraints $G^X_f$ and $G^Z_f$ simultaneously, making it difficult to study. There are many different Hamiltonians whose ground spaces contain the codespace of the gauge color code, and in the next subsection we introduce a solvable model, consisting of mutually commuting terms. 

\subsection{A commuting model}

Here we define an exactly solvable model for the gauge color code. The Hamiltonian is given by a sum of gauge terms that belong to 3-cells of a single color. Without loss of generality, fix this color to be \textbf{b} (blue), and take all faces $X_f$ and $Z_f$ belonging to the blue 3-cells or blue boundary facet. That is, all faces $f$ that have color $\textbf{uv}$ with $\textbf{u}, \textbf{v} \neq \textbf{b}$. Label the set of these faces by 
\begin{equation}
\calG_{\textbf{b}} = \{G^X_f, G^Z_f ~|~ \calK (f)  \in \{ \textbf{gr}, \textbf{gy}, \textbf{ry} \}\},
\end{equation}
where $\calK (f)$ denotes the color of $f$. Note that $\calG_{\textbf{b}}$ consists of commuting terms, as all terms are supported on either a bulk 3-cell or the \textbf{b} boundary (which are both 3-colorable and 3-valent sublattices). Or equivalently, if two faces share a common color then the terms commute. We can define an exactly solvable Hamiltonian by 
\begin{equation}
H_{\calG_{\textbf{b}}} = -\sum_{G \in \calG_{\textbf{b}}} G.
\end{equation}
This Hamiltonian decomposes into a number of decoupled 2D color codes, one on the \textbf{b} boundary, and one for each bulk 3-cell of color \textbf{b}.  Additionally, every qubit is in the support of at least one $G \in \calG_{\textbf{b}}$. 

With the above choice of boundary conditions, the outer colex (the \textbf{b} boundary) encodes one logical qubit, while the bulk 2D color codes are non-degenerate (as they are each supported on closed 2-cells). The ground space of the model is the joint $+1$ eigenspace of all terms $G\in \calG_{\textbf{b}}$, and the ground space degeneracy is two-fold. This choice of Hamiltonian explicitly represents the gauge color code codespace on the outer colex. This situation is reminiscent of the RBH model, where quantum information is encoded on the boundary of the 3D bulk. We remark that the ground state of $H_{\calG_{\textbf{b}}}$ can be thought of as a gauge fixed version of the gauge color code $\calG$.

Logical operators can be chosen to be string-like operators supported entirely on the outer colex (the \textbf{b} boundary facet). Recall that edges and plaquettes on the outer colex has one of three possible colors, $\textbf{A}$, $\textbf{B}$, or $\textbf{C}$, as defined in Fig.~\ref{fig2dCC}, and the boundaries are given a single color according to what color edges can terminate on them, as depicted in Fig.~\ref{figBoundaryFacet}. The logical operators take the form of strings that connect all three boundaries of the triangular facet as in Fig.~\ref{figBoundaryFacet}. Logical Pauli operators are supported on at least $d$ qubits, where $d$ is the smallest side length of the boundary facet and referred to as the distance of the code. 

\begin{figure}[h]
	\centering
	\subfloat[]{
	\includegraphics[width=0.45\linewidth]{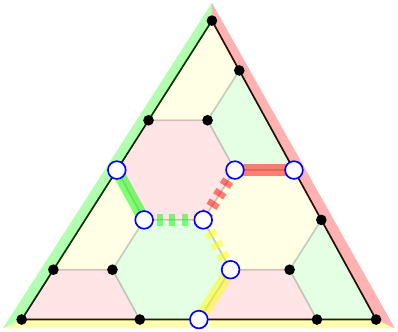}}
	\quad
	\subfloat[]{
	\includegraphics[width=0.45\linewidth]{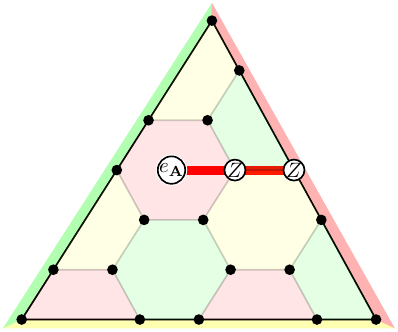}}
	\caption{(a) A logical string consists of three colored strings extending from their respective boundary and meeting at a point. The support of the logical $X$ or $Z$ is indicated by the larger white nodes. (b) $e_{\textbf{A}}$ excitations appear at the ends of a \textbf{A}-colored $Z$-string. Both $\textbf{e}_{\textbf{A}}$ and $\textbf{m}_{\textbf{A}}$ excitations can condense on the \textbf{A}-colored boundary (and analogously for other boundaries).}
	\label{figBoundaryFacet}
\end{figure}

On the outer colex, an $X$- or $Z$-string operator with color $\textbf{k} \in \{\textbf{A}, \textbf{B},\textbf{C}\}$ will flip the two $\textbf{k}$ coloured plaquettes on the boundary of the string. In particular, a $\textbf{k}$-colored $X$-string will create  $m_{\textbf{k}}$ excitations on its boundary (corresponding to the flipped $G_f^Z$ plaquettes). Similarly, a $\textbf{k}$-colored $Z$-string will create  $e_{\textbf{k}}$ excitations on its boundary (corresponding to the flipped $G_f^X$ plaquettes). These are depicted in Fig.~\ref{figBoundaryFacet}. On a $\textbf{k}$ colored boundary, both $e_{\textbf{k}}$ and $m_{\textbf{k}}$ particles can condense, meaning they can be locally created or destroyed at the boundary as in Fig.~\ref{figBoundaryFacet}. 

As such, the action of logical $\overline{X}$ ($\overline{Z}$) can then be interpreted as creating three $m$-type ($e$-type) quasiparticles of each color from the vacuum at a point, then moving each colored excitation to its like-colored boundary where it is destroyed. 

\subsubsection{Relation to the RBH model}

To motivate how the model $H_{\calG_{\textbf{b}}}$ was constructed, we draw a comparison to the RBH model of the previous section. In particular, the RBH also has the structure of a subsystem code, that on a certain lattice is dual to the gauge color code. For the RBH model, one can consider the gauge group $\calG_{\calC}$ is given by
\begin{equation}
\calG_{\calC} = \langle K_p , X_p ~|~ p \in E \cup F \rangle,
\end{equation}
where $K_p$ are the cluster state stabilizers of Eq.~(\ref{eqClusterStateStabilizers}) and $X_p$ are single qubit Pauli $X$ operators.  The corresponding stabilizer group $\calS_{\calC}$ is given by 
\begin{equation}
\calS_{\calC} = \langle S_p ~|~ p \in Q \cup V \rangle,
\end{equation}
where $S_p$ are the 1-form symmetry generators of the RBH model, given by Eq.~(\ref{eqRBH1formSymOps}).
(The choice of gauge generators $X_p$ stems from the application of the RBH model to fault-tolerant measurement-based quantum computing, where $X$-measurements are used to propagate information.)

The commuting model describing the RBH model was chosen by selecting a subset $\calG'$ of local, commuting elements of $\calG_{\calC}$ to define the Hamiltonian, and imposing symmetries given by the stabilizer $\calS_{\calC}$. This choice is non-unique, as there are many other subsets $\calG'$ of $\calG$ that could be used to construct a commuting model. Additionally, to avoid spontaneous symmetry breaking we choose $\calG'$ such that the stabilizer is a subgroup of the group generated by $\calG'$, that is, $\calS_{\calC} \leq \langle \calG' \rangle$. The same construction was also used to generate the commuting GCC model, and can be used more generally for subsystem codes with a stabilizer group that has the structure of a $\zz_2^k$ 1-form symmetry for some $k$. We note however there are many distinct ways generating such Hamiltonians, and not all of them will be self-correcting under the 1-form symmetry.

\subsection{1-form symmetry and color flux conservation}
\label{sec:colorflux}

The commuting model $H_{\calG_{\textbf{b}}}$ without any symmetry constraints is easily shown to be disordered at any non-zero temperature. (It is a collection of uncoupled 2D color codes.) In this section, we identify a 1-form symmetry of this model that, when enforced, leads to a diverging energy barrier and therefore self-correction on the boundary code.

The Hamiltonian $H_{\calG_{\textbf{b}}}$ has a $\zz_2^2$ 1-form symmetry given by the stabilizer group $\calS$ of Eq.~(\ref{eqGCCStabilizer}). Recall that $\calS$ is generated by the stabilizers $S_q^X$ and $S_q^Z$ on the 3-cells $q$ of the lattice, and consists of operators supported on closed codimension-1 (contractible) surfaces. The two copies of $\zz_2$ 1-form symmetry come from the independent $X$-type and $Z$-type operators. The symmetry $\calS$ give strong constraints (conservation laws) on the possible excitations in the model: this is the color flux conservation of Bombin~\cite{Bombin15}. To discuss the color flux conservation that arises from the $\zz_2^2$ 1-form symmetry, let us assume that the system $H_{\calG_{\textbf{b}}}$ is coupled to a thermal bath (as in Eq.~(\ref{eqCoupleBath})) such that the whole system respects the symmetry $\calS$, and discuss what type of excitations are possible in the model.

The model $H_{\calG_{\textbf{b}}}$ is a stabilizer Hamiltonian, and so excitations are labelled in the standard way.  Specifically, excited states can be labelled by the set of `flipped terms' $\calG_{\text{ex}} \subseteq \calG_{\textbf{b}}$. Not all sets $\calG_{\text{ex}}$ can be reached from the ground space in the presence of the symmetry $\calS$. Since the ground space of $H_{\calG_{\textbf{b}}}$ consists of the states in the $+1$ eigenspace of all terms in $\calG_{\textbf{b}}$, it follows that the ground space is also the $+1$ eigenspace of all operators in $\calS$, and since they are conserved, only the excited states that satisfy color flux conservation on each cell (as we will describe) can be reached. 

In particular, note that for any 3-cell $q$ of color $\textbf{k} \neq \textbf{b}$, there is precisely one way of obtaining the stabilizers $S_q^X$ and $S_q^Z$ from terms in $\calG_{\textbf{b}}$, while for a 3-cell of color $\textbf{b}$ there are three ways of obtaining the stabilizers. More precisely, for the $X$-type stabilizers we have 
\begin{align}\label{eqGCCRigid}
S_q^X = \prod_{\substack{f \subset q \\ \calK(f) = \textbf{uv}}} G_f^X,
\end{align}
where 
\begin{align} \label{eqColorConstraints}
\textbf{uv} &\in  \begin{cases}
\{\textbf{gy}\} \quad \text{if } \calK(q) = \textbf{r} \\ 
\{\textbf{ry}\} \quad \text{if } \calK(q) = \textbf{g} \\ 
\{\textbf{rg}\} \quad \text{if } \calK(q) = \textbf{y} \\ 
\{\textbf{gy}, \textbf{ry}, \textbf{rg}\} \quad \text{if } \calK(q) = \textbf{b}.
\end{cases}
\end{align}
The above expression holds similarly for the stabilizer $S_q^Z$. This can be seen as any plaquette that neighbours a 3-cell of color $\textbf{k}$ must be of color $\textbf{uv}$ with $\textbf{u}, \textbf{v} \neq \textbf{k}$, for which there is only one choice within $\calG_{\textbf{b}}$ for $\textbf{k} \neq \textbf{b}$, and three choices when $\textbf{k} = \textbf{b}$. Note that the multiple ways of forming $S_q^X$ and $S_q^Z$ on blue 3-cells as per Eq.~(\ref{eqColorConstraints}) leads to local product constraints on these blue 3-cells (further constraining the excitations) however this is not important for the present discussion.

To ensure that an excitation $\calG_{\text{ex}}$ is valid, we must remain in the $+1$-eigenspace of $\calS$. From Eq.~(\ref{eqGCCRigid}) we see that every 3-cell $q$ must have an even number of flipped plaquettes belonging to its boundary. Indeed, a single flipped plaquette $G_f^X$ of color $\textbf{uv}$ would violate the two stabilizer operators $S_q^X$ and $S_{q'}^X$ on the neighbouring \textbf{u} and \textbf{v} colored 3-cells $q$ and $q'$. This constraint implies that symmetric excitation configurations consist of collections of closed loop-like sets of flipped plaquettes.

This can be more easily visualised on the dual lattice, where where 3-cells are replaced by vertices, faces by edges, edges by faces, and vertices by 3-cells. On the dual lattice, vertices carry a single color, edges are labelled by pairs of colors, and excitations are therefore given by sets of edges. We call the edges on the dual lattice that define an excitation a flux string. The color flux conservation on these closed flux strings is as follows.

To satisfy the constraints of Eqs.~(\ref{eqGCCRigid}), and (\ref{eqColorConstraints}), for each vertex $v$ of color $\textbf{k}\in \{\textbf{b}, \textbf{r}, \textbf{g}, \textbf{y}\}$ the number of edges in a flux string incident to $v$ must be even. Since the vertices of color $\textbf{k}\in \{\textbf{r}, \textbf{g}, \textbf{y}\}$ only support terms in $\calG_{\textbf{b}}$ on neighbouring edges of a single color type (e.g. a $\textbf{r}$ vertex only supports terms on its neighbouring $\textbf{gy}$-colored  edges), then the color of the excitation is conserved at each one of these vertices. Similarly on a \textbf{b} vertex, all pairs of colors are separately conserved. This means if a $\textbf{uv}$ colored edge excitation enters a vertex, there must be a $\textbf{uv}$ colored edge excitation leaving the vertex. In summary, bulk excitations must form closed loops, where the color is conserved at every vertex, and this is illustrated in Fig.~\ref{figBoundaryFlux}. 

\begin{figure}[h]
	\subfloat[]{	\includegraphics[width=0.48\linewidth]{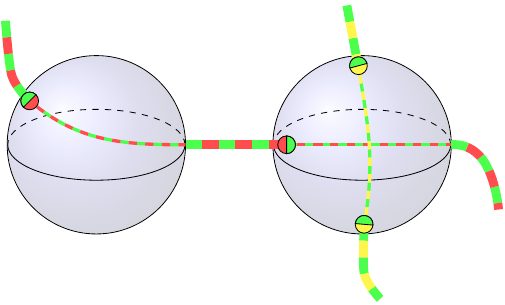}}
	\subfloat[]{	\includegraphics[width=0.46\linewidth]{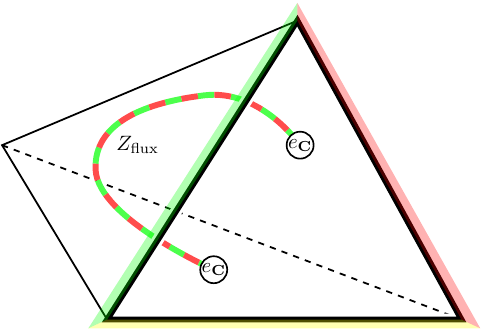}}
	\caption{(a) An example of a flux loop, where the corresponding colored strings on the dual lattice are depicted, the shaded blue spheres represent \textbf{b} colored 3-cells. (The constraint from Eq.~(\ref{eqGCCRigid}) requires an even number of flipped $\textbf{rg}$ plaquettes on a \textbf{y} colored 3-cell). (b) A \textbf{rg} colored flux loop of flipped $G_f^Z$ terms (coming from a string of $X$ operators) terminating with a pair of $e_{\textbf{C}}$ anyons on the outer colex.}
	\label{figBoundaryFlux}
\end{figure}

Flux loops may terminate on the outer colex. Recall that for a boundary facet of color $\textbf{k}$, there are no faces of color $\textbf{uk}$ for any $\textbf{u}$. In particular, for $\textbf{k}\neq \textbf{b}$, there is a unique color \textbf{u} such that there are terms $G_f^X$ and $G_f^Z$ of color $\textbf{uk}$ in $\calG_{\textbf{b}}$. Flux loops of color \textbf{uk} can terminate on this \textbf{k}-colored boundary facet. For the \textbf{b} colored boundary facet (the outer colex), all three color pairs of flux loops can terminate on the outer colex. Flux loops terminating on the \textbf{b}-facet can be viewed as ending in a $e_{\textbf{k}}$ or $m_{\textbf{k}}$ anyonic excitation on the boundary for $\textbf{k}\in \{\textbf{A}, \textbf{B}, \textbf{C}\}$ as in Fig.~\ref{figBoundaryFlux} (recall the colors are relabelled on the outer colex according to $\textbf{gy}\leftrightarrow \textbf{A}$, $\textbf{ry}\leftrightarrow \textbf{B}$, $\textbf{rg}\leftrightarrow \textbf{C}$). Moreover, in the same way, the only way anyons can exist on the outer colex is at the ends of a flux loop on the bulk, as stand-alone boundary anyonic excitations violate the symmetry. That is, the 1-form symmetry couples the bulk and boundary excitations, as was the case in the RBH model. 

\subsection{Energy barrier}
We are now equipped to calculate the symmetric energy barrier for $H_{\calG_{\textbf{b}}}$ in the presence of the symmetry $\calS$. Recall that a logical error occurs when a triple of excitations $\alpha_{\textbf{A}}, \alpha_{\textbf{B}}, \alpha_{\textbf{C}}$, where $\alpha = e$ or $m$, are created at a point, and each anyon travels to its like-colored boundary. Put another way, a logical error occurs if an anyonic excitations $\alpha_{\textbf{k}}$ is created at each boundary, and the three anyons move and fuse back to the vacuum in the bulk of the outer colex. In any case, the only way to achieve a logical Pauli error is to create a number of anyonic excitations, which must move a combined distance of at least $d$, the side length of the outer colex.  In the symmetric sector, anyonic excitations can only exist on the boundary if they are accompanied by a bulk flux loop, and so the above creation, movement and fusion process can only occur when accompanied by bulk flux loops. 

Since boundary excitations $\alpha_{\textbf{k}}$ with $\alpha \in \{e,m\}$ and $\textbf{k} \in \{\textbf{A},\textbf{B},\textbf{C}\}$ appear on the end of flux loops (each of which can only terminate on its like-colored boundary) to calculate the energy barrier we need only track the smallest length flux loops required to move the boundary anyons to create a logical error. From any point $v$ on the outer colex, let $l_{\textbf{A}}(v)$, $l_{\textbf{B}}(v)$, $l_{\textbf{C}}(v)$ be the shortest flux loops from a face $f$ on the outer colex containing $v$, to a face on the \textbf{A}, \textbf{B}, and \textbf{C} facets, respectively (these flux loops are dual to a closed path on the dual lattice). Let $|l_{\textbf{A}}(v)|$, $|l_{\textbf{B}}(v)|$, $|l_{\textbf{C}}(v)|$ be the lengths of these flux loops (i.e., the number of edges on the dual path) and define 
\begin{equation}
d_{\perp} := \min_{v \in \calC_{\text{out}}} \left(|l_{\textbf{A}}(v)|+|l_{\textbf{B}}(v)|+|l_{\textbf{C}}(v)| \right)
\end{equation}
to be the shortest combined distance from any point on the outer colex to all three other facets. Note that $d_{\perp}$ grows as all side lengths of the tetrahedral 3-colex are increased. 

Then during any anyon creation, movement and annihilation process resulting in a logical error, the bulk flux loops which accompany the boundary anyons must have a combined length of at least $d_{\perp}$. This will incur an energy penalty of $\Delta_E = 2d_{\perp}$ since each flux loop consists of a path of flipped terms $G_f^{\alpha} \in \calG_{\textbf{b}}$. As such the energy is proportional to $d_{\perp}$ which scales linearly with the minimum side length of the tetrahedral 3-colex. In particular, the model $H'$ with symmetry $\calS$ has a macroscopic energy barrier, and the boundary information is protected in the presence of a 3D bulk and symmetry constraint. 

We make two remarks. First, the energy barrier and conservation laws in this section were presented in terms of excitations rather than error operators (as opposed to the operator approach for the RBH model). For the purposes of calculating the energy barrier these two pictures are equivalent, since the sequence of local (symmetric) excitations corresponds to a sequence of local (symmetric) operators, and vice-versa. Second, we remark that a tri-string logical operator of the above form can be pushed onto a single boundary of the outer colex, giving rise to a string-like representative. As such, a logical error can arise from a pair of anyons of the same color being created and moved along the boundary of the outer colex. Such a process also has an energy lower bounded by $\Delta_E = 2d_{\perp}$ since a \textbf{k}-colored string on the boundary of the outer colex is never adjacent to a boundary where its \textbf{k}-flux loops can terminate. 

The argument from the symmetric energy barrier to self-correction follows identically to that of the RBH model. That is, provided the temperature is sufficiently low, information can be stored for a time that grows exponentially with the system size.  (Note that the critical temperature will depend on the specific choice of 3-colex.) As a result, our stabilizer model based on the 3D gauge color code protected by $\zz_2^2$ 1-form symmetry provides another example of a self-correcting quantum memory. 

In the RBH model, the fact that the boundary was self-correcting in the presence of 1-form symmetries could be interpreted as directly resulting from the thermally stable bulk SPT order. In this stabilizer model of the gauge color code, the boundary stability and bulk SPT (at nonzero temperature) are also related~\cite{KYgauging}.

\section{Emergent 1-form symmetries}\label{secEmergentHigherForm}

As we have shown, SET models protected by a 1-form symmetry can be self-correcting.  However, enforcing such 1-form symmetries is a very strong constraint, and in addition these symmetries are unusual in physics compared with the more prevalent onsite (0-form) symmetries.  Here we explore the idea that 1-form symmetries may actually appear naturally in 3D topological models, and not require any sort of external enforcement.  We refer to such a symmetry as \emph{emergent}.  It sounds too good to be true, but note that emergent symmetries in 2D topological models are ubiquitous (while perhaps poorly understood).  In this section, we review emergent (0-form) symmetries in 2D topological models, as first highlighted by Kitaev~\cite{Kitaev2003}; here we will focus on the 2D color code.  We then show that 3D models may possess emergent 1-form symmetries associated with such emergent 0-form symmetries on closed 2D submanifolds of the 3D model.  We revisit the 3D gauge color code in light of these observations.  Finally, we demonstrate the stability of emergent 1-form symmetries in topologically ordered models, and discuss the implications for self-correction.

We note that this section contains arguments that are less formalized compared with the previous sections, and in many regards more speculative.  As such, this section may be viewed as an extended discussion on the potential role of emergent 1-form symmetries in self-correction, rather than the presentation of concrete results.

\subsection{Emergent 0-form symmetries in 2D}

Kitaev observed the emergence of symmetry in 2D topological models such as the toric code and referred to this as a `miracle'~\cite{Kitaev2003}.  As we now know, emergent symmetries are a generic property of 2D topologically ordered models.  We begin this section by reviewing an instructive first example: the 2D color code.  We demonstrate the emergence of a $\zz_2^4$ 0-form symmetry in this 2D code, and how this gives rise to the well known anyonic color conservation (see for example Ref.~\cite{YoshidaCCSPT}). Although we will focus on how global product constraints are helpful to expose global conservation laws, we emphasise that the more important physical property is the local conservation law (associated with a 0-form symmetry) that arises in relation to the modular Gauss law.

We first consider a 2D color code defined on the surface of a sphere (one can equivalently consider any closed surface for the discussion that follows). Recall, a 2D color code is defined on a lattice known as a 2-colex, which is a 3-colorable, 3-valent cellulation $\Lambda$ of a 2-dimensional surface, which in this case is a sphere. We place a qubit on each vertex of $\Lambda$, and define the familiar $X$-type and $Z$-type face operators $G^X_f= \prod_{v \in f} X_v$ and $G^Z_f = \prod_{v \in f} Z_v$ for each face $f \subset \Lambda$. In particular, since the lattice is 3-colorable and 3-valent, these face operators $G^X_f$ and $G^Z_f$ all commute. These operators generate the 2D color code stabilizer group $\calS_{cc} = \langle G^X_f, G^Z_f ~|~ f \text{ a face of } \Lambda \rangle$, and define a corresponding Hamiltonian $H_{\text{2D-cc}}$ by
\begin{equation}
H_{\text{2D-cc}} = -\sum_{\text{faces } f} \left(G^X_f + G^Z_f\right). 
\end{equation}
This 2D color code differs only from that defined on the outer colex (considered in Sec.~\ref{sec:GCC}, Fig.~\ref{fig2dCC}) by a choice of boundary conditions. 

Recall, a generating set for the anyonic excitations of this model can be labelled by $m_{\textbf{k}}$, and $e_{\textbf{k}}$, where $\textbf{k}\in \{\textbf{A},\textbf{B}\}$ labels a color, $e$-type anyons corresponds to flipped $X$-type plaquettes, and $m$-type anyons correspond to flipped $Z$-plaquettes. One can obtain $\textbf{C}$ colored anyons as the fusion of an $\textbf{A}$ and $\textbf{B}$ colored anyon of the same type. This set of anyons forms a group under fusion $\mathcal{A}_{\text{2D-cc}}\cong \zz_2^4$, with the above choice of generators. 

However, not all anyonic excitation configurations are possible as there are global constraints that need to be satisfied in this model. In particular, since our model is defined on a closed surface, we have the following identities for each $\alpha \in \{X, Z\}$
\begin{equation}
\prod_{\substack{f \subset \Lambda \\ \calK(f) = \textbf{A}}} G^{\alpha}_f =  \prod_{\substack{f \subset \Lambda \\ \calK(f) = \textbf{B}}} G^{\alpha}_f = \prod_{\substack{f \subset \Lambda \\ \calK(f) = \textbf{C}}} G^{\alpha}_f = \prod_{v \in \Lambda} \alpha_v.
\end{equation}
Letting $N_{\textbf{k}}^e$ and $N_{\textbf{k}}^m$ be the number of $e_{\textbf{k}}$ and $m_{\textbf{k}}$ anyonic excitations respectively, then the above equation implies the following relation 
\begin{equation}
N_{\textbf{A}}^e = N_{\textbf{B}}^e = N_{\textbf{C}}^e \quad \mod 2,
\end{equation}
and similarly for $N_{\textbf{k}}^m$. In particular this means that the number of $e_{\textbf{A}}$, $e_{\textbf{B}}$ and $e_{\textbf{C}}$ anyons is conserved mod 2 (and similarly for $m_{\textbf{A}}$, $m_{\textbf{B}}$, and $m_{\textbf{C}}$). 

If we regard anyons of color $\textbf{C}$ as being comprised of an $\textbf{A}$ color and a $\textbf{B}$ color anyon, we can obtain further constraints. Namely, for any two colors, $\textbf{u}, \textbf{v} \in \{\textbf{A}, \textbf{B}, \textbf{C}\}$, we have a product constraint
\begin{equation}\label{eq2dccconstraints}
\prod_{\substack{f \subset \Lambda \\ \calK(f) = \textbf{u}}} G^{\alpha}_f  \prod_{\substack{f \subset \Lambda \\ \calK(f) = \textbf{v}}} G^{\alpha}_f = I. 
\end{equation}
This implies a constraint on the parity of anyons
\begin{equation}\label{eqParityConstraint1}
N_{\textbf{u}}^e + N_{\textbf{v}}^e = 0 \quad \mod 2,
\end{equation}
which along with the fact that we are regarding $N_{\textbf{C}}^e = N_{\textbf{A}}^e + N_{\textbf{B}}^e$, means that $N_{\textbf{A}}^e = N_{\textbf{B}}^e = 0$ mod $2$ (and similarly for $m$-type anyons). The product constraint of Eq.~(\ref{eq2dccconstraints}) exists on the whole 2-dimensional lattice (that is, a codimension-0 surface), and gives rise to 4 independent anyonic constraints: that the number of $e_{\textbf{A}}$ anyons must be created or destroyed in pairs, and similarly for $e_{\textbf{B}}$, $m_{\textbf{A}}$, and $m_{\textbf{B}}$. Thus, we refer to it as an emergent $\zz_2^4$ 0-form symmetry. 

The identities of Eq.~(\ref{eq2dccconstraints}) make this emergent symmetry look like a global constraint, however it is in fact a 0-form symmetry. That is, we can identify an action of this symmetry on any submanifold, not just the whole lattice.  This structure to the symmetry is best seen by reformulating it as a type of Gauss' law for anyonic excitations, detecting the total topological charge in a region through an observable localized to the boundary of the region.  Specifically, consider submanifolds that are not closed. Let $\calM$ be a codimension-0 submanifold of the 2-colex (that is, a subset of faces), with boundary. Then for $\alpha \in \{X,Z\}$ it holds that 
\begin{equation}\label{eqEmergBdry}
\prod_{\substack{f \subset \calM \\ \calK(f) = \textbf{u}}} G_f^{\alpha} \prod_{\substack{f \subset \calM \\ \calK(f) = \textbf{v}}} G_f^{\alpha}  = h_{\partial M},
\end{equation}
where $h_{\partial M} = \prod_{v \in \partial M} \alpha_v$ is supported on the boundary of $\calM$. (Note that we have assumed the 2-colex is closed, however the above equation also holds when $\calM$ is disjoint from the boundary of the 2-colex). Now instead of the global constraint of Eq.~(\ref{eqParityConstraint1}), we get a constraint for every submanifold $\calM$. Namely, the charge within the region $\calM$ is equal (mod 2) to the eigenvalue of on the operator $h_{\partial M}$
\begin{equation}\label{eqParityConstraint}
N_{\textbf{u}}^e + N_{\textbf{v}}^e = \langle h_{\partial \calM}\rangle \quad \mod 2,
\end{equation}
for any excited state (provided, as is true with this model, that anyons are well-localized). Choosing $\textbf{v}= \textbf{C}$ lets us determine $N_{\textbf{A}}^e$ and $N_{\textbf{B}}^e$ independently, and similarly for $N_{\textbf{k}}^m$. In other words, one can detect the topological charge within the region $\calM$ using operators on the boundary of the region, giving rise to the well-known topological charge conservation law for anyons in the color code. Thus we have seen that the conservation law applies locally as well (provided that the length scale is such that anyons remain well-localized), and is not just a global constraint on the entire manifold.

Importantly, in the above considerations, emergent symmetries were revealed not by elements of a symmetry group, but rather by product constraints amongst the Hamiltonian terms.  This is a result of the stabilizer Hamiltonian models that we have considered as examples.  We can now turn to higher-dimensional examples, again of stabilizer Hamiltonians, where this holds true for higher-form symmetries, i.e., where emergent $q$-form symmetries are associated with product constraints on closed codimension-$q$ submanifolds of the lattice.  Ultimately, however, we expect the symmetry considerations rather than the product constraints to be more fundamental, and we return to this issue in Sec.~\ref{secEmergentGeneral}.

\subsection{Emergent 1-form symmetries in 3D}

Here we demonstrate how emergent 1-form symmetries can arise in a 3D model, in a sense by bootstrapping from the 2D case.  

\subsubsection{Single-sector 3D gauge color code}

For illustrative purposes, we first consider a single charge sector of the 3D gauge color code $H_{\calG}$.  This single-sector model is not topologically ordered, and so does not possess emergent symmetries; nonetheless it will be useful to illustrate the connection between 1-form symmetries in a 3D model and 0-form symmetries in associated 2D models existing across all codimension-1 submanifolds of the 3D model. The 1-form symmetries fix excitations to be 1-dimensional objects that conserve color flux.

Recall, the gauge color code is defined on a 3-colex $\calC_3$ (a 4-colorable, 4-valent cellalation) with a qubit on each vertex. For concreteness, we restrict our discussion to the $X$-sector of the gauge color code (the $Z$-sector follows similarly). That is, we consider the Hamiltonian 
\begin{equation}
H_X = -\sum_f G^X_f,
\end{equation}
consisting of the sum of all face terms over a 3-colex. The ground space of $H_X$ is the mutual $+1$ eigenspace of all terms $G^X_f$, and excitations are eigenstates of the Hamiltonian in the $-1$ eigenspace of some terms (we say these terms are $G^X_f= -1$). We can label excited states uniquely by specifying which terms are $G^X_f = -1$, but importantly not all configurations are allowed, as there are algebraic constraints amongst terms. 

Consider any closed codimension-1 submanifold $\calM$ of the 3-colex that is also a 2-colex, with the color-pairs $\textbf{A}_{\calM}$, $\textbf{B}_{\calM}$, and $\textbf{C}_{\calM}$ selected from the 6 possible color-pairs of faces in $\calC_3$. On this sub-2-colex, we have the familiar constraints. Namely, for any 2 color-pairs $\textbf{u}, \textbf{v} \in \{\textbf{A}_{\calM}, \textbf{B}_{\calM},\textbf{C}_{\calM}\}$, we have 
\begin{equation}\label{eqSub3DConstraints}
\prod_{\substack{f \subset \calM \\ \calK(f) = \textbf{u}}} G^{X}_f  \prod_{\substack{f \subset \calM \\ \calK(f) = \textbf{v}}} G^{X}_f = I, 
\end{equation}
mirroring the constraints of Eq.~(\ref{eq2dccconstraints}). In particular, this relation holds in the smallest instance when $\calM$ is the boundary of a 3-cell.

The product relations of Eq.~(\ref{eqSub3DConstraints}) lead to constraints on excitations. Namely, for each codimension-1 submanifold (that is a 2-colex), the number of faces $f \subset q$ with $G^X_f = -1$ carrying a color \textbf{k} must sum to $(0\mod~2)$, and this holds for each (single) color \textbf{k}. This in turn requires excitations (which carry pairs of colors) to form closed loop-like objects that conserve color. The dual lattice again provides the visualization, where excitations correspond to sets of edges and edges carry a pair of colors. At each vertex $v$ of the dual lattice, let $N^{v}_{\textbf{k}}$ be the number of loop excitations carrying the (single) color $\textbf{k}$ that contain $v$. Then the constraints of Eq.~(\ref{eqSub3DConstraints}) mean that 
\begin{equation}
N^{v}_{\textbf{k}} = 0\,, \quad \forall\ \textbf{k}, v\,,
\end{equation}
which is precisely the color flux conservation discussed in Sec.~\ref{sec:colorflux}. In particular, this implies that excitations must form closed loop-like objects.

Not all excitations are independent. A string excitation of a color \textbf{xz} may branch into a pair of strings with colors \textbf{xk} and \textbf{kz} for $\textbf{k} \neq \textbf{x}, \textbf{z}$. This then means there are three independent color pairs, such that all loop excitations can be regarded as the fusion of these loops. The flux conservation can be regarded as three independent constraints on loop-like excitations.  

Similar to the 0-form case, 1-form symmetries also imply a constraint (conservation law) for the loop-like excitations. We can infer a generalization of the law for detecting topological charge, which in this case applies to color flux, by considering codimension-1 submanifolds that are not closed. In particular, let $\calM'$ be a codimension-$1$ submanifold with a boundary. Then it holds that 
\begin{equation}\label{eqEmergBdry2}
\prod_{\substack{f \subset \calM' \\ \calK(f) = \textbf{u}}} G^{X}_f  \prod_{\substack{f \subset \calM' \\ \calK(f) = \textbf{v}}} G^{X}_f = h_{\partial M'},
\end{equation}
where $h_{\partial M'}$ is an operator supported on the (1-dimensional) boundary of $\calM$ (again we are assuming that $\calM$ is supported away from any boundary of the 3-colex). This means that the number (mod 2) of $\textbf{u}$ colored and $\textbf{v}$ colored excitations that thread the region $\calM'$ is detected by an operator $h_{\partial M'}$ on the boundary of that region. Again, we can use the constraints to determine this number on each independent color pair. 

In summary, we have seen that this model supports three independent types of excitations, each constrained to form closed loops (with the possibility of branching and fusion). This 3D example, then, gives the appearance of an emergent $\zz_2^3$ 1-form symmetry arising from a 0-form symmetry on codimension-1 submanifolds (where the rank of the 1-form symmetry group is due to the number of independent excitations that are conserved). We note, however, that by restricting to the $X$-sector, we do not have a topologically ordered model; the codimension-1 submanifolds do not have an \emph{emergent} $0$-form symmetry without both sectors, and so an emergent 1-form symmetry does not appear in the 3D model.  Both electric and magnetic sectors are required simultaneously in order to have the emergent symmetry associated with either~\cite{Kitaev2003}.  Regardless, our purpose here was simply illustrative---we are not fundamentally interested in this single-sector model, but rather a topologically-ordered 3D model with both sectors such as the gauge color code.  We turn to that model now.

\subsubsection{The gauge color code and color flux conservation}

Does the topologically-ordered 3D gauge color code have an emergent 1-form symmetry associated with color flux conservation?  Each sector of the gauge color code on its own, $H_X$ and $H_Z$, has loop-like, color-flux-conserving excitations. Proliferation of such excitations is therefore suppressed, as they are energetically confined. 
For the full gauge color code Hamiltonian, 
\begin{equation}\label{eqGCCHam2}
H_{\calG} = -\sum_f G^X_f - \sum_f G^Z_f,
\end{equation}
it is tempting to conclude that a $\zz_2^6$ 1-form symmetry will emerge, and lead to confined errors and suppression of logical faults. However, the terms of $H_{\calG}$ are not mutually commuting (and indeed frustrated), and therefore we cannot immediately label excited states by specifying terms $G_f^X, G_f^Z = \pm1$.  In other words, this frustrated model's excitations are not guaranteed to be well-defined extended objects with well-defined color flux as appear in each sector separately. If they were, then this would be strong evidence that the model was self-correcting. 

Unfortunately, there are few tools available to understand the spectrum of a frustrated Hamiltonian such as $H_{\calG}$, and without such information it is a very difficult task to analyse the thermal stability and memory time of the code. In this sense, one can view the exactly solvable model $H_{\calG_{\textbf{b}}}$ as the result of removing terms from the Hamiltonian until it is commuting, in the process losing its emergent 1-form symmetries and supplementing them with enforced 1-form symmetries. Understanding the excitations in $H_{\calG}$ remains an important problem, to determine if it is self-correcting. 

\subsubsection{Higher-dimensional generalizations and emergent $q$-form symmetries}

We briefly generalize the discussion to emergent $q$-form symmetries in $d$-dimensional systems that arise from (product) constraints residing on codimension-$q$ submanifolds. In particular, a commuting Hamiltonian $H = \sum_{X \subset \Lambda} h_X$ in $d$-dimensions has an emergent $\zz_2$ $q$-form symmetry if for all closed codimension-$q$ submanifolds $\calM$, there exists an constraint
\begin{equation}\label{eqEmerSym}
\prod_{X \subset \calM}h_X = I.
\end{equation}
If there are multiple independent such constraints on the submanifolds, then there are multiple copies of emergent $\zz_2$ $q$-form symmetries. Importantly, we note that these constraints all look like emergent $\zz_2$ 0-form symmetries on codimension-$q$ submanifolds. The generalized conservation law states that the number (mod 2) of excitations (which must be $q$-dimensional objects) threading the codimension-$q$ region $\calM'$ can be measured by the operator $H_{\partial M'}$ on the codimension-$(q{+}1)$ boundary of the region. In particular, if $H$ has a $q$-form emergent symmetry, let $\calM'$ be a codimension-$q$ submanifold with a boundary, then it holds that 
\begin{equation}\label{eqEmergBdry3}
\prod_{i \in \calM'} h_i = h_{\partial M'}\ ,
\end{equation}
where $h_{\partial M'}$ is an operator supported on a small neighbourhood of the boundary of $\calM$. (This is because if we chose a complementary codimension-$q$ submanifold $\calM''$ such that $\partial M' = \partial M''$, then if $\calM$ is the result of gluing $\calM$ and $\calM'$ along their boundary, we would have the usual constraint of Eq.~(\ref{eqEmerSym}). Thus $\prod_{i \in \calM'} h_i$ can only differ from the identity by an operator supported on a small neighbourhood of $\partial \calM'$.)

Examples of models with emergent higher-form symmetries include toric codes in various dimensions. For dimensions $d\geq2$, there are $d{-}1$ distinct ways of defining a toric code. Namely, for each $k \in \{1,\ldots, d{-}1\}$, we define the $(k,d{-}k)$ toric code that has $k$-dimensional logical $X$ operators, and $(d{-}k)$-dimensional logical $Z$ operators. One can confirm that these models have emergent $\zz_2$ $(k{-}1)$-form and $\zz_2$ $(d{-}k{-}1)$-form symmetries. The smallest dimension that allows for a toric code with emergent $\zz_2^2$ 1-form symmetries is $d=4$, with the $(2,2)$ toric code, which is a self-correcting quantum memory.

\subsection{Stability of emergent symmetries}\label{secEmergentGeneral}

Our discussion of emergent symmetries has focussed on Hamiltonians with commuting terms.  This property allowed for the simple identification of product constraints.  One can ask if the resulting emergent symmetries are a property of a finely tuned system alone, or if they hold more generally.  In this section, we show that these symmetries are robust features of phases of matter, that they cannot be broken by local perturbations, \emph{irrespective of any symmetry considerations}, provided they are sufficiently small. The argument uses the idea of quasi adiabatic continuation, following Ref.~\cite{QuasiCont}. 

Consider a family of local Hamiltonians $H_s$, labelled by a continuous parameter $s \in [0,1]$, such that $H_0 = H$ is the original Hamiltonian, and $H_s$ remains gapped for all $s \in [0,1]$. This family of Hamilonians can be used describe the situation where a perturbation is added to $H$. We label ground states of $H$ by $\ket{\psi_i}$, and groundstates of $H_s$ by $\ket{\psi_i^s}$. Note that the ground states can be unitarily related by an adiabatic continuation. Then, following Ref.~\cite{QuasiCont}, there exists a unitary $U(s)$ corresponding to a quasi-adiabatic change of the Hamiltonian with the following properties. For any operator $O$, one can find a dressed operator $O_s = U(s) O U(s)^{\dagger}$, such that $O_s$ has approximately the same expectation value in $\ket{\psi_{i}^s}$ as $O$ does in $\ket{\psi_i}$ (and similarly for low-energy states). Moreover, if $O$ is local, then $O_s$ is local too.  (The support of the dressed operators increases by a size determined by the choice of quasi-adiabatic continuation unitary $U(s)$. The approximate ground state expectation values improve exponentially in the range of increased support of dressed operators.)

Importantly, one can use quasiadiabatic continuation to find dressed versions $h_X(s) = U(s) h_X U(s)^{\dagger}$ of the Hamiltonian terms that have approximately the same low-energy expectation values as those in the unperturbed Hamiltonian. These Hamiltonian terms will also have the same constraints. In particular, if $H$ had an emergent $q$-form symmetry arising from some product constraints amongst Hamiltonian terms, then the dressed Hamiltonian also has the same local conservation laws. To see this, note that local conservation laws can always be inferred at low energies, as they involve only Hamiltonian terms in a small neighbourhood. We needn't be concerned with the high energy sector as by checking all local conservation laws, one can establish that the model has an emergent $q$-form symmetry. Note that the dressed terms will in general be supported in a larger region, meaning one may need to rescale the lattice to resolve excitations and faithfully capture the generalized conservation law in the perturbed Hamiltonian. For example, consider the color code in the presence of perturbations, then one can renormalize the lattice such that individual excitations are well defined. Then in the renormalized lattice, these excitations still conserve anyon parity, and they still obey a conservation law for topological charge.

We remark that we required the gap to remain open in the presence of the perturbations. This can be guaranteed for any local perturbation (provided it is sufficiently weak), if $H$ satisfies the conditions of TQO-1 and TQO-2 of Ref.~\cite{TopoStability}. In particular, the example models we have considered in Sec.~\ref{sec:RBH} and Sec.~\ref{sec:GCC} satisfy the conditions. 

\subsection{Duality between emergent and enforceable symmetries}
\label{sec:Duality}

For emergent symmetries, we are faced with the puzzle that we have a conservation law without any symmetry operator. What is the origin of this symmetry? As pointed out by Kitaev in the case of the 2D toric code~\cite{Kitaev2003}, we can always recover symmetry operators by introducing redundant ``unphysical'' degrees of freedom, viewed as gauge degrees of freedom. Here we briefly consider how Kitaev's approach can be applied to higher-form symmetries. In particular, for systems with emergent symmetries, we will construct symmetry operators on an enlarged Hilbert space. This construction provides a duality between systems where the $q$-form symmetry is emergent and systems where it is enforced. 

We will begin with the color code in 2D, and then show how to lift the construction to the 1-form case in 3D.  We start by introducing new ancillary degrees of freedom---one ancilla for each term in the Hamiltonian. Label these ancilla by $a_X(f)$ and $a_Z(f)$ corresponding to the terms $G^X_f$ and $G^Z_f$ and fixed them in the $+1$ eigenspace of Pauli operators $X$ and $Z$, respectively. We can now regard the new Hilbert space as $\calH \otimes \mathcal{A}$, and states in $\calH$ are embedded according to the isometry $\ket{\psi} \mapsto \ket{\psi}\otimes \ket{a}$, where $\ket{a} = (\otimes_{a_X(f)} \ket{+}) (\otimes _{a_Z(f)} \ket{0})$. We refer to the (original) degrees of freedom in $\calH$ as matter, and those in $\mathcal{A}$ as gauge. Importantly, not all states $\ket{\phi}\in\calH\otimes\mathcal{A}$ are physical, only the subspace of states satisfying $X_{a_X(f)}\ket{\phi} = \ket{\phi}$ and $Z_{a_Z(f)}\ket{\phi} = \ket{\phi}$ are physical. At this point, it is clear from the embedding that the physical state space is the same as the original state space. 

We now couple the matter and gauge degrees of freedom with an entangling unitary. Consider the mapping of gauge terms and matter Hamiltonian terms
\begin{align}
X_{a_X(f)} &\mapsto S_f^X, & G^X_f&\mapsto G^X_f,\\
Z_{a_Z(f)} &\mapsto S_f^Z, & G^Z_f&\mapsto G^Z_f,
\end{align}
where $S^X_{f} = X_{a_X(f)} G^X_f$ and $S^Z_{f} = Z_{a_X(f)} G^Z_f$. Such a mapping can be achieved with a unitary $U$ as we show below. In this new Hilbert space, which we label $U(\calH\otimes\mathcal{A})U^{\dagger}$, the physical state space is the subspace satisfying 
\begin{equation}\label{eqGaugeSym}
S^X_{f}\ket{\phi} = S^Z_{f} \ket{\phi} = \ket{\phi}.
\end{equation}
The symmetry operators $S^X_{f}$ and $S^Z_{f}$ are known as gauge transformations, and states and operators that are related by them are thought of as equivalent. 

The entangling unitary $U$ that will result in the above mapping can be constructed out of 2-qubit CNOT gates, $A_{i,j}$, which act by conjugation on Pauli operators as follows 
\begin{align}
X_i &\mapsto X_i X_j, & Z_i &\mapsto Z_i\\
X_j &\mapsto X_j &  Z_j &\mapsto Z_iZ_j.
\end{align}
Then for each face $f$, we define the following unitaries
\begin{equation}\label{eqGaugeUnitary}
U_f^X = \prod_{v \in f} A_{a_X(f), v}, \quad U_f^Z = \prod_{v \in f} A_{v, a_Z(f)}.
\end{equation}
Note that $U_f^X$ has the following action: 
\begin{equation}
U_{f'}^X X_{a_X(f)} U_{f'}^{X\dagger}  = \begin{cases}
S_f^X \quad  &\text{if } f = f' \\
X_{a_X(f)} \quad &\text{otherwise}.
\end{cases}
\end{equation}
Moreover, $U_{f'}^X$ commutes with all Hamiltonian terms $G^X_f$ and $G^Z_f$ $\forall f$ (this statement only needs to be verified for terms $G^Z_f$ where $f'$ and $f$ are neighbours, where it holds because neighbouring terms intersect an even number of times -- as is always the case for commuting CSS stabilizer Hamiltonians). A similar calculation gives the action of $U_f^Z$
\begin{equation}
U_{f'}^Z Z_{a_Z(f)} U_{f'}^{Z\dagger}  = \begin{cases}
S_f^Z \quad  &\text{if } f = f' \\
Z_{a_Z(f)} \quad &\text{otherwise}
\end{cases}
\end{equation}
where again $U_{f'}^Z$ commutes with all Hamiltonian terms $G^X_f$ and $G^Z_f\ \forall f$. Then the desired unitary $U$ is given by $U = \prod_{f} U_{f}^X U_{f}^Z$.  

Since the Hamiltonian is unchanged by $U$, one can ask what the excitations in the physical space of $U(\calH\otimes\mathcal{A})U^{\dagger}$ look like. Namely, for each flipped term $G^X_f$ ($G^Z_f$) we must also flip the ancilla $a_X(f)$ ($a_Z(f)$). Thus one can equally label excitations by the terms $G^X_f$ and $G^Z_f$, or the terms $X_{a_X(f)}$ and $Z_{a_Z(f)}$, as the two sets are gauge equivalent.  The emergent 0-form symmetry manifests itself as product constraints amongst Hamiltonian terms (following Eq.~(\ref{eq2dccconstraints})).  Specifically, it is equivalent to the following constraints, for any color $\textbf{u} \neq \textbf{v}$
\begin{equation}\label{eqGauge0Sym}
\prod_{\substack{f |\calK(f) = \textbf{u}}} S^X_{f} \prod_{\substack{f | \calK(f) = \textbf{v}}} S^X_{f} = \prod_{a_X(f)} X_{a_X(f)},
\end{equation}
and similarly for the $Z$-terms. Here, we see that the operator $\prod_{a_X(f)} X_{a_X(f)}$ (which is gauge equivalent to a product of color code terms $G_f^X$) counts the number of excitations mod 2.  As it is a product of symmetry operators, any physical state must lie in its $+1$ eigenspace. That is, we have found a symmetry operator that determines the parity conservation of anyons, by introducing gauge degrees of freedom. 

In the same way, we can perform an analogous procedure for each sector in the 3D gauge color code. Again, we associate ancilla to each term in the Hamiltonian, and then apply the unitary $U$ that entangles gauge and matter degrees of freedom. Much like the 2D case, this leads to symmetry operators constructed on all codimension-1 submanifolds (out of products of $S_f^X$ and $S_f^Z$ on these surfaces) and a requirement that the physical states must live in their common $+1$ eigenspace (the enforced 1-form symmetry). These symmetry operators mirror the 1-form operators that we have seen in sections \ref{sec:RBH} and \ref{sec:GCC}.  In fact, this construction works for any CSS stabilizer code (in any dimension), where the product over $v\in f$ in Eq.~(\ref{eqGaugeUnitary}) is replaced by product over the qubits in the support of the stabilizer term. 

By introducing redundant degrees of freedom, we have related a model with an emergent symmetry to one with an enforced symmetry.  The duality mapping known as \emph{gauging}~\cite{levin2012braiding,williamson2017symmetry,YoshidaHigher,haegeman2015gauging,hung2012string,williamson2016matrix} formalizes this relationship. In the abelian case, gauging a model with an onsite (0-form) symmetry can produce a model with an emergent 0-form symmetry. For example, the 2D toric code contains 0-form emergent symmetries and can be obtained by gauging the trivial paramagnet.  Gauging also provides a potential direction for identifying models with emergent 1-form symmetries.  We note that formalisms for gauging/ungauging more general types of symmetries have been explored by Vijay, Haah, and Fu~\cite{vijay2016fracton}, Williamson~\cite{williamson2016fractal}, as well as Kubica and Yoshida~\cite{KYgauging}; these approaches provide potentially powerful tools to identify self-correcting quantum memories protected by emergent 1-form symmetries.

We also remark on the parallels between this simple duality mapping and error correction. In fact, the coupling of gauge degrees of freedom is similar to many schemes of syndrome extraction, where measurement of ancillas is used to infer the eigenvalues of stabilizer terms.  Measurement errors can break this correspondence, however, and result in a misidentification of errors.  This is typically accommodated by requiring many rounds of measurements.  For single shot error correction (such as in the GCC \cite{SingleShot}), only a single round of measurements is needed, owing to the extensive number of symmetry constraints present, whose violation indicates a measurement error. In the case of emergent 0-form symmetries, the global constraint alone cannot provide sufficient information to correct for measurement errors. In a similar vein to self-correction in 3D, it would be interesting find 2D topological codes (if they exist) with emergent $\zz_2^2$ 1-form symmetries, as such codes could in principle admit single-shot error correction. 
 
\section{Discussion}
\label{sec:Discussion}

We have shown that spin lattice models corresponding to 2D SET ordered boundaries of thermally-stable 3D SPT ordered phases protected by a suitable 1-form symmetry can be self-correcting quantum memories. The key features of these 1-form symmetric models are that the bulk excitations are string-like and confined, and that the symmetry naturally couples bulk and boundary excitations to confine the later as well. 
	
We have presented two explicit examples of 3D self-correcting quantum memories protected by 1-form symmetries.  The understanding and classification of such 3D models remains largely unexplored.  A natural class of candidates are the (modular) Walker-Wang models~\cite{BulkBoundary1,Walker2012,vonKeyserlingk13,Burnell2013,Wang2017}, which possess many of the desirable properties we seek.  In particular, if the input anyon theory to the Walker-Wang construction is modular, then all bulk excitations are confined, while the 2D boundary contains a copy of the input anyon theory. One can consider building 1-form symmetries into these types of models, as has been done by Williamson and Wang~\cite{williamson2017hamiltonian} for a class of models based on the state sum TQFTs of Ref.~\cite{cui2016higher}.  (We note this is similar to the way that Ref.~\cite{BulkBoundary1} `decorates' a Walker-Wang model with a 0-form symmetry.) The 2-group construction of Ref.~\cite{kapustin2017higher} presents another interesting family of models that warrants further investigation.  In the stabilizer case, another possible approach to construct 3D models with 1-form symmetries is to ``foliate''~\cite{bolt2016foliated, brown2018universal} a topological stabilizer code with emergent 0-form symmetries.  As an example, foliation of a $d$-dimensional topological CSS code with emergent $q$-form symmetry generates a $(d{+}1)$-dimensional generalized RBH-type model with a $(q{+}1)$-form symmetry. A rigorous classification of all boundaries of such 1-form SPT phases -- including the RBH model -- remains an interesting open problem.

In the examples we have explored, we have seen the necessity of the bulk SPT-ordering in order to have a self-correcting boundary, and for the bulk SPT-ordering of these models to be thermally stable. A common viewpoint is that a self-correcting quantum memory should be topologically ordered at nonzero temperature. While this has not been proven to be strictly necessary, it has been observed to be true for many examples under Hastings' definition for topological order at $T\geq 0$ \cite{hastings2011topological}. (For example, 2D commuting projector Hamiltonian models and the 3D toric code all lack topological order at $T\textgreater0$, corresponding to the absence of self-correction.) Our examples provide further support to this perspective.

We briefly consider what our results imply for self-correction in the 3D gauge color code.  As we have shown in Sec.~\ref{sec:GCC}, the 3D gauge color code realized as commuting Hamiltonians protected by an (enforced) 1-form symmetry is self-correcting.  If we consider the full Hamiltonian of Eq.~(\ref{eqGCCHam}), the model is frustrated and it is difficult to prove that it possesses the string-like excitations with well-defined topological charge required for our arguments.  We have also argued that the full model possesses an emergent 1-form symmetry:  the color flux conservation as previously identified by Bombin~\cite{Bombin15}.  This emergent symmetry gives strong supporting evidence that proving self-correction for the full Hamiltonian of Eq.~(\ref{eqGCCHam}) (without enforcing any symmetry requirement) may be possible. What remains is to understand the spectrum of the model, and in particular verify whether the energy cost of a loop excitation grows with its length.

The idea that 1-form symmetries may be emergent in 3D topological models is extremely intriguing, both from the perspective of self-correction and more generally.  We have argued that 1-form symmetries may emerge in 3D models that possess emergent 0-form symmetries on all codimension-1 submanifolds, which in turn can be guaranteed by topological ordering of these submanifolds.  We can ask whether the 1-form symmetries of the RBH model or commuting GCC model can be realised in an emergent fashion in a 3D commuting, frustration-free Hamiltonian.  It is not clear if this is possible.  The key goal here is to identify models that possess well-defined bulk excitations together with sufficient emergent 1-form symmetries to guarantee confinement for all of such excitations. This is in contrast to the 3D toric code, where only one sector has an emergent 1-form symmetry, and correspondingly only one type of logical operator is thermally stable (giving rise to a self-correcting classical memory). Topological subsystem codes, such as the gauge color code, are natural candidates. Along with obviating the need to enforce symmetries, another advantage of emergent symmetries is that the conservation laws are manifestly true, without putting any restrictions on the system-bath coupling. 

A key open question is how to construct more general families of models with emergent higher-form symmetries. We have discussed a simple duality between emergent and enforceable symmetries, that symmetries can be introduced by adding gauge degrees of freedom in systems with emergent symmetries. In the case of 0-form symmetries, a simple well-known gauging map~\cite{levin2012braiding,williamson2017symmetry,YoshidaHigher,haegeman2015gauging,hung2012string,williamson2016matrix} can be used to obtain a model with emergent $\zz_2$ 0-form symmetry from a model with an enforced $\zz_2$ 0-form symmetry. Investigating this more generally in the presence of both enforced and emergent higher-form symmetries may lead to interesting new models, and here we point the interested reader to new results by Kubica and Yoshida on generalized gauging and ungauging maps~\cite{KYgauging}.

We have not considered the issue of efficient decoding for these self-correcting quantum memories.  We note that our two examples, the RBH model and the gauge color code, have efficient decoders with the additional feature of being single-shot~\cite{RBH,ThermalSPT,brown2016fault}.  In general, we note that the string-like nature of the excitations (errors) in these 1-form symmetric self-correcting quantum memories ensure that efficient decoders exist in general~\cite{KubicaToAppear}.

Finally, there are many avenues for further investigation into the role of symmetry in self-correcting quantum memories. In particular, one can consider the stability and feasibility of self-correction in defect-based encodings, for example in twist defects~\cite{bombin2010topological,brown2017poking} or the ``Cheshire charge'' loops of Refs.~\cite{else2017prethermal,else2017cheshire}. Such defects have a rich connection with SPT order, as well as with both enforced and emergent symmetries. Namely, as shown in Ref.~\cite{YoshidaCCSPT}, one can view topological phases with nontrivial domain walls as having SPT ground states protected by 0-form symmetries, where the protecting symmetry comes from the emergent 0-form symmetries of the topological model. It would be interesting to see if SPTs protected by higher-form symmetries also arise in this way, that is, from domain walls of topological models with emergent higher-form symmetries, and whether these associated domain walls (and symmetry defects that live on their boundaries) can be thermally stable. For example, the SPT order (at temperature $T\geq 0$) in the RBH model manifests as a thermally stable domain wall in the 4D toric code~\cite{ThermalSPT}. Whether one can construct similarly stable domain walls in 3D or less is an open problem. Another direction is to consider more general subsystem symmetries, where the dimension need not be an integer. For example, fracton topological orders (which can be partially self-correcting \cite{Haah}) have been of great interest  recently~\cite{vijay2016fracton,williamson2016fractal,ma2017fracton}.

\begin{acknowledgments}
We thank Aleksander Kubica and Beni Yoshida for discussions, and for kindly sharing drafts of their work on ungauging quantum error correcting codes~\cite{KYgauging}. We also thank Ben Brown, Tomas Jochym-O'Connor, Andrew Doherty and Dom Williamson for helpful discussions.  This work is supported by the Australian Research Council (ARC) via the Centre of Excellence in Engineered Quantum Systems (EQuS) project number CE170100009 and Discovery Project number DP170103073.
\end{acknowledgments}

%

\newpage
\appendix

\section{Davies Formalism}\label{sec:Davies}
In this appendix we briefly review the Davies formalism. Recall the system-bath coupling
\begin{equation}
H_{\text{full}} = H_{\text{sys}} +H_{\text{bath}} +  \lambda \sum_{\alpha} S_{\alpha} \otimes B_{\alpha},
\end{equation}
where $S_{\alpha} \otimes B_{\alpha}$ describe the system-bath interaction for $S_{\alpha}$ a local operator acting on the system side, $B_{\alpha}$ is an operator acting on the bath side, and $\alpha$ is an arbitrary index. It is assumed that the coupling parameter is small, $|\lambda|\ll 1$.
Suppose that the state is initialized in a ground state $\rho(0)$ of $H_{\text{sys}}$, then the state evolves under a Markovian master equation 
\begin{equation}\label{eqLindblad}
\dot\rho(t) = -i[H_{\text{sys}}, \rho(t)] + \calL(\rho(t)),
\end{equation}
where $\calL$ is the Lindblad generator. Then the initial ground state $\rho(0)$ evolves under this master equation according to 
\begin{equation}
\rho(t) = e^{t\calL} (\rho(0)).
\end{equation}
Here, the Lindblad generator is given by
\begin{equation}\label{eqLindblad2}
\calL(\rho) = \sum_{\alpha, \omega} h(\alpha, \omega) \left(A_{\alpha, \omega} \rho A_{\alpha, \omega}^{\dagger} - \frac{1}{2}\{\rho, A_{\alpha, \omega}^{\dagger} A_{\alpha,\omega}\}.\right)
\end{equation}

In the above, $A_{\alpha,\omega}$ are the Fourier components of $A_{\alpha}(t) = e^{iH_{\text{sys}} t} A_{\alpha} e^{-i H_{\text{sys}} t}$, meaning they satisfy
\begin{equation}
\sum_{\omega}e^{-i \omega t} A_{\alpha, \omega} = e^{iH_{\text{sys}} t} A_{\alpha} e^{-i H_{\text{sys}} t}.
\end{equation}
One can think of $A_{\alpha, \omega}$ as the component of $A_{\alpha}$ that transfers energy $\omega$ from the system to the bath. Note that when the Hamiltonian $H_{\text{sys}}$ is comprised of commuting terms, the terms $A_{\alpha}(t)$ and therefore also $A_{\alpha, \omega}$ are local operators. The function $h(\alpha, \omega)$ can be thought of as determining the rate of quantum jumps induced by $A_{\alpha}$ that transfer energy $\omega$ from the system to the bath, and is the only part that depends on the bath temperature. It must satisfy the detailed balance condition $h(\alpha, -\omega) = e^{-\beta \omega} h(\alpha, \omega)$, which ensures that the Gibbs state
\begin{equation}
\rho_{\beta} = e^{-\beta H_{\text{sys}}}/ \Tr(e^{-\beta H_{\text{sys}}}),
\end{equation}
at inverse temperature $\beta$ is a fixed point of the dynamics of Eq.~(\ref{eqLindblad}). That is, $\rho_{\beta} = \lim_{t \rightarrow \infty} \rho(t)$. Moreover, under natural ergodicity conditions (see \cite{Ergodic1,Ergodic2} for more details), it is the unique fixed point. 

In the case that we have a symmetry ,
\begin{equation}\label{eqSBSym2}
[H_{\text{full}}, S(g)] = 0,
\end{equation}
then all of the errors that are introduced due to interactions with the bath must be from processes that conserve $S(g)$. In particular, only excitations that can be created by symmetric thermal errors will be allowed. Indeed, in the case that Eq.~(\ref{eqSBSym2}) holds, we will have that 
\begin{equation}
e^{\calL t} (S(g)^{\dagger} \rho_0 S(g)) = S(g)^{\dagger} e^{\calL t} (\rho_0) S(g)
\end{equation}
which justifies the consideration of the symmetric energy barrier in Eq.~(\ref{eqSymEnergyBarrier}).

We note that the assumptions of this formalism are satisfied for systems where the terms are comprised of commuting Paulis, as in this case the system Hamiltonian has a discrete spectrum with well separated eigenvalues. However the formalism will not necessarily work beyond this exact case, for instance, when perturbations are added and small energy splittings are introduced between previously degenerate eigenvalues. The study of thermalization times for many body stabilizer Hamiltonians in the presence of perturbations is an interesting problem.  

\section{Thermal instability of 0-form SPT ordered memories}\label{sec:ThermalOnsite}

In this appendix we argue that onsite symmetries are insufficient to promote a 2D topological quantum memory to be self-correcting, even if such a phase lives on the boundary of a 3D SPT model.  We restrict our discussion to the case where the boundary Hamiltonian is an abelian twisted quantum double.  The interesting case is where the boundary symmetry action is anomalous.  (However we don't allow this boundary symmetry action to permute the anyon types.)

We will argue that the boundary theory of a 3D SPT ordered bulk phase, if topologically ordered, will necessarily possess deconfined anyons.  That is, the boundary string operators corresponding to error chains can be deformed while still respecting the symmetry, even with anomaly.   We focus on (twisted) quantum doubles on the boundary of 3D group cohomology SPTs, and rather than going into the details of their construction, we focus on the key features. In particular, local degrees of freedom (of both bulk and boundary) for these models are labelled by group elements, as $\ket{g}$, $g \in G$. The symmetry action of these 2D (boundary) systems takes the form $S(g) = R(g) N(g)$, where $R(g) = \otimes_{i} u(g)$, with $u(g)= \sum_{h\in G} \ket{gh}\bra{h}$ and $N(g)$ is diagonal in the $\ket{g}$ basis and can be represented as a constant depth quantum circuit. One can think of $R(g)$ as the onsite action, and $N(g)$ as an anomaly.  This anomaly must be trivial in a strictly 2D system, or equivalently if the system is at the boundary of a trivial SPT phase.  

There are two types of excitation operators in the (twisted) quantum doubles. One type of excitation string operator for the boundary system is diagonal in the $\ket{g}$ basis (i.e., it is the same as in the untwisted theory), so it commutes with $N(g)$. This excitation string operator commutes with $u(g)$, up to a phase (that is a $k$th root of unity for some $k\in \mathbb{N}$), so to commute with $R(g)$ we need to consider excitation string operators of certain lengths. In particular, the process of creating an anyonic excitation at one boundary and dragging it to another boundary (or creating an anyon pair and dragging one around a nontrivial cycle before annihilating them again) can be done in a symmetric way. Since such an operation results in a logical error and only costs a constant amount of energy, we see that the boundary theory is unstable. 

Thus we see that the anomaly affords no extra stability, and the model has the same stability as a topological model with an extra onsite symmetry on top. That is, like genuine 2D topological models of this type, the model has a constant symmetric energy barrier. Note that this argument can break down in 4D, where the boundary is a 3D twisted quantum double.

Therefore we see that in the case of onsite (0-form) symmetries, the SPT ordered bulk offers no additional stability to the boundary theory. Indeed, the  symmetric energy barrier for the abelian twisted quantum double remains the same as the energy barrier without symmetry:  constant in the size of the system.  This motivates us to consider the boundaries of SPTs protected by 1-form (or other higher-form) symmetries.

\section{Short-range insensitivity of the boundary properties}\label{UVInsensitivity}

We show that boundaries of the 1-form symmetric RBH phase in Secs.~\ref{ToricBoundaries} and~\ref{secOtherBoundaries} are not sensitive to short-range (UV) details of the lattice.
	
Consider a lattice with boundary neighbourhood $R$ and Hamiltonian on this boundary $H_{R}$. We consider the case where $H_{R}$ is known, for example it may be the toric code boundary Hamiltonian of Sec.~\ref{ToricBoundaries}. Consider erasing a region $\mathcal{E} \subset R$, in the sense of removing all Hamiltonian terms with support intersecting $\mathcal{E}$, along with potentially adding or removing spins. We show that knowledge of the Hamiltonian restricted to $R\setminus \mathcal{E}$, which we denote $H_{R\setminus \mathcal{E}}$, is sufficient to determine the possible symmetric Hamiltonians $H_{\mathcal{E}}$ on $\mathcal{E}$ and that they must belong to the same phase of matter as  $H_{R\setminus \mathcal{E}}$. In other words, knowledge of $H_{R}$ within some region (with a potentially nice lattice geometry) is sufficient to canonically choose the Hamiltonian in neighbouring regions which may not enjoy the nice lattice geometry.

Our argument extends upon that of Sec.~\ref{SecSymBoundary}, where we determine the boundary theory by analysing the loop-like boundary action of 1-form symmetry operators. All 2D topological order can be understood in this way -- in terms of certain loop operators that preserve the ground space. In the case of stabilizer models, this property can be formulated in a precise way~\cite{bombin2012universal}. Indeed, for stabilizer models, the topological order is entirely characterised by the local commutation relations between loops that stabilize the ground space. 

Consider a pair of 1-form symmetry operators $S_1$, $S_2$ with some support on the boundary $R$. The action of these 1-form operators on the boundary degrees of freedom is given by loop-like operators $l_1$, $l_2$. Loop operators $l_1$ and $l_2$ must preserve the ground space of the Hamiltonian $H_{R}$ on the boundary in order for the boundary to be symmetric. We consider the case where the induced loop operators $l_1$ and $l_2$ intersect within two regions $a$ and $b$ with $a\subset R\setminus \mathcal{E}$ and $b\subseteq \mathcal{E}$. Since $[S_1, S_2] = 0$ we must also have $[l_1, l_2] = 0$. From this it follows $[l_1|_a, l_2|_a] = 0 $ $\iff $ $[l_1|_b, l_2|_b] = 0 $ and $\{l_1|_a, l_2|_a\} = 0 $ $\iff $ $\{l_1|_b, l_2|_b\} = 0 $, where the notation $|_{a}$ and $|_{b}$ denotes the restrictions of the operators to $a$ and $b$, respectively. That is, knowledge of local commutation relations of loops within $a$ is sufficient to determine the local commutation relations within $b$. As all 2D stabilizer models are characterised by the the local commutation relations of loops that stabilize their ground space~\cite{ bombin2012universal}, the phase in $R\setminus\mathcal{E}$ is sufficient to determine the phase within the erased region $\mathcal{E}$. 

Our discussion is simplified greatly by the fact that our model is a topological stabilizer model, but similar arguments can be made for more general 1-form symmetric models in 3D. We remark that our argument bears some resemblance with Hasting's ``healing the puncture" technique in Ref.~\cite{hastings2013classifying}.

\section{Energy barrier is sufficient}\label{PeierlsR}

In this appendix, we consider the timescale for logical faults in the 1-form symmetric RBH model.  We estimate the probability that an excitation loop $l$ of size $w$ emerges within the Gibbs ensemble at inverse temperature $\beta$. We show that large loop errors are quite rare if the temperature is below a critical temperature $T_c$, which we lower bound by $2/\log(5)$.
 
Recall the symmetric excitations are given by applying operators $Z(E', F') = \prod_{f \in F'}Z_f \prod_{e \in E'}Z_e$, where $E'$ is a cycle (i.e.,~has no boundary) and $F'$ is dual to a cycle on the dual lattice. We will refer to both such subsets $E'$ and $F'$ as cycles, $l = E' \cup F'$, and the resulting excitation $\ket{\psi(l)}$ as an excitation loop configuration. Moreover, we will refer to each connected component of $l$ as a loop (intuitively loops are minimal in that no proper subset of a loop can be a cycle). The energy $E(\gamma)$ of such an excitation configuration is given by $2|(E' \cup F') \cap \interior{\calL}| + 2|\partial (E' \cup F') \cap \partial {\calL}|$, i.e., it is proportional to the length of the bulk cycle plus the number of times a bulk cycle touches the boundary. Then the Gibbs state $\rho_{\beta}$ is given by the weighted mixture of all symmetric excitations, where the weights are given by
\begin{equation}
P_{\beta}(\gamma) = \frac{1}{\calZ}e^{{-\beta E{(\gamma)}}}, \quad \calZ = \sum_{\gamma} P_{\beta}(\gamma),
\end{equation}
and $\gamma = (E', F')$ represents a valid (i.e.,~symmetric) excitation. 

Define $d = \min \{d_{Z} , d_{X} , d_{\text{cond}} \}$ from Def.~\ref{defndist}. For a logical error to have occurred during the system-bath interaction, we must pass through an excited state $\ket{\psi(c)}$ such that $c$ contains a bulk loop with length $w \geq d - r$, for some constant $r$ independent of system size. (Here a bulk loop is one where at least half of its support is away from the boundary). Let us bound the probability that configurations containing such a loop occurs. Define $\calB_w$ to be the set of cycles containing a bulk loop with size at least $w$. Then 
\begin{align}
\sum_{c \in \calB_w} P_{\beta}(c) &\leq \sum_{\substack{\text{loops } l \\ |l| \geq w}} ~\sum_{\substack{\text{cycles } c \\ l \subset c}} P_{\beta}(c) \\ 
&\leq \sum_{\substack{\text{loops } l \\ |l| \geq w}} ~e^{-\beta E(l)}\sum_{\substack{\text{cycles } c \\ l \not\subset c}} P_{\beta}(c) \\ 
&\leq \sum_{\substack{\text{loops } l \\ |l| \geq w}} ~e^{-\beta E(l)}, \label{eqLoopBound} 
\end{align}
where from the first to the second line we have used that a configuration $c$ containing a loop $l$ differs in energy from the configuration $c \setminus l$ by $E(c) = e^{-\beta E(l)} E(c\setminus l)$. Now the last line can be rewritten to give  
\begin{align}
\sum_{c \in \calB_w} P_{\beta}(c) &\leq \sum_{k \geq w} N(k) e^{-2\beta k},
\end{align}
where we have ignored contributions to $E(l)$ due to the boundary (these will only decrease the right hand side of Eq.~(\ref{eqLoopBound})) and $N(k)$ counts the number of loops of size $k$. Since a loop $l$ resides on either the primal or dual sublattice, each of which has the structure of a cubic lattice, we can obtain a crude upper bound on $N(k)$ by considering a loop as a non-backtracking walk, where at each step one can move in 5 independent directions. This gives the bound $N(k) \leq p(d) 5^m=k$, where $p(d)$ is a polynomial in $d$, and is in particular proportional to the number of qubits. 

Then, provided $T \leq 2/ \log(5)$, we have 
\begin{align}
\sum_{c \in \calB_w} P_{\beta}(c) &\leq p(d) \sum_{k \geq w} e^{k(\log(5)- 2\beta)}\\
&= p(d)\frac{1}{(1-e^{\log(5)- 2\beta })} e^{k(\log(5)- 2\beta)}
\end{align}
which is exponentially decaying in $k$ (again provided $T \leq 2/ \log(5)$). Since errors can be achieved only if we pass through a configuration with a bulk loop of length $d-r$, we have the contribution of configurations that can cause a logical error is bounded by 
\begin{equation}
\text{poly}(d)\frac{1}{(1-e^{-\alpha})} e^{-\alpha d}
\end{equation}
where $\alpha = 2\beta - \log(5) \textgreater 0$ is satisfied when the temperature is small enough. One can show that the decay rate of the logical operators is exponentially long, and therefore the fidelity of the logical information is exponentially long in the system size (see Proposition 1 of Ref.~\cite{alicki2010thermal}). One could perform a more detailed calculation to show that, with a suitable decoder, error correction succeeds after an evolution time that grows exponentially in the system size (i.e., that logical faults are also not introduced during the decoding).

We also note that a similar argument can be made for the commuting gauge color code model of Sec.~\ref{sec:GCC}. A different critical temperature will be observed that depends on the choice of 3-colex.

\end{document}